\documentclass[12pt]{article}
\usepackage{amsmath}
\usepackage{graphicx}
\usepackage{enumerate}
\usepackage{natbib}
\usepackage{bm}
\usepackage{url} 

\usepackage{amsthm}
\usepackage{amssymb}

\usepackage{xr}
 
\usepackage{scalerel}
 
\usepackage{xcolor}
\usepackage{multirow}
\usepackage{arydshln}
 
\usepackage{soul}

\usepackage{float}

\usepackage{latexsym}
\usepackage{caption}
 
\usepackage{mathtools}

\theoremstyle{definition}
\newtheorem{assumption}{Assumption}

\newtheorem*{theorem*}{Theorem}
\newtheorem{theorem}{Theorem}
\newtheorem*{rmk*}{remark}

\newtheorem{proposition}{Proposition}
\newtheorem{lemma}{Lemma}

\newtheorem{definition}{Definition}
\newtheorem{remark}{Remark}

\newtheorem*{corollary*}{Corollary}

\usepackage{etoolbox} 
\apptocmd{\sloppy}{\hbadness 10000\relax}{}{} 

\usepackage{multibib}
\newcites{sec}{References}

\usepackage{color}
\usepackage{listings}
\usepackage{hyperref}
\hypersetup{
  colorlinks=true,
  linkcolor=blue,
  citecolor=blue,
  urlcolor=blue
}

\usepackage{booktabs}

\usepackage{lscape}
\graphicspath{ {./Graphs/} }

\RequirePackage[normalem]{ulem}

\usepackage{setspace}

\allowdisplaybreaks

 \def\spacingset#1{\renewcommand{\baselinestretch}%
{#1}\small\normalsize}

\newcommand{\txi}{\tilde X_i}
\newcommand{\xidm}{X_{i,\textup{dm}}}
\newcommand{\txidm}{\tilde X_{i, \textup{dm}}}
\newcommand{\txif}{(X_{i1}, \ldots, X_{i,K-1})^\T}
\newcommand{\txit}{\txi^\T}

\newcommand{\hmuk}{\hat \mu_k}

\newcommand{\bdiff}{\beta_{\textup{diff}}}
\newcommand{\pcd}{\pr\left(\uicd\right)}\newcommand{\cxi}{\check X_i}
\newcommand{\assmivnomono}{\assmiv\eqref{it:indep}--\eqref{it:relevance}}\newcommand{\mx}{\mathcal X}
\newcommand{\propv}{\prop~\ref{prop:v}}
\newcommand{\bdx}{\beta_{DX}}
\newcommand{\bdv}{\beta_{DV}}
\newcommand{\tlx}{\tl(\xxi)}
\newcommand{\tcx}{\tc(\xxi)}

\newcommand{\bd}{\beta_\dd}
\newcommand{\ud}{\ui = \dd}
\newcommand{\smf}{strong monotonicity}
\newcommand{\wmf}{weak monotonicity}

\newcommand{\thmbi}{\thm~\ref{thm:bi}}\newcommand{\tl}{\tau_{\textsc{late}}}
\newcommand{\thmnece}{\thm~\ref{thm:nece}}

\newcommand{\defli}{Definition~\ref{def:li}}
\newcommand{\uicd}{\ui \in \{\cc,\dd\}}
\newcommand{\ui}{U_i}

\newcommand{\uic}{\ui = \cc}

\newcommand{\assmivps}{\assm~\ref{assm:ivps}}

\newcommand{\assmiv}{\assm~\ref{assm:iv}}
\newcommand{\assmwm}{\assm~\ref{assm:wm}}

\newcommand{\exi}{e(X_i)}
\newcommand{\otn}{i = \ot{N}}\def\tslst{\texttt{2sls}}

\def\condivf{{the product of the IV propensity score and the covariates is linear in the covariates}}

\def\pc{\pi_\cc}

\def\cdi{\hat D_i}
\def\taa{\tau_{++}}
\def\bcv{\bc'}

\def\xio{X_{i1}}

\def\vlem{Assume $X_i = (V_i^\T, W_i^\T)^\T \in \mathbb R^{K}$, where $V_i \in \mathbb R^Q$ with $Q\leq K$ and $\E(\xxit)$ is invertible}
\def\cq{takes exactly $Q$ distinct values}
\def\ck{takes exactly $K$ distinct values}

\def\xxt{X\xt}
\def\xxit{X_i\xit}
\def\xxit{X_i\xit}
\def\vvit{V_i\vit}

\def\xidm{X_{i,\textup{dm}}}
\def\xt{X^\T}
\def\exx{\E(\xxt)}
\def\exxt{\E(\xxt)}
\def\exxinv{\left\{\E(\xxt)\right\}^{-1}}
\def\dxi{D_iX_i}
\def\zxi{Z_iX_i}

\def\zvi{Z_iV_i}
\def\dvi{D_iV_i}

\def\vit{V_i^\T}

\def\bv{\bi'}
\def\hbv{\hbi'}

\def\its{interacted \tsls}

\newcommand{\lmt}{\texttt{ols}}
\newcommand{\tsls}{\textsc{2sls}}

\newcommand{\yid}{Y_{i}(d)}
\newcommand{\wi}{W_i}

\def\dti{\delta_i}
\def\ydset{\{\yizd, \dizz:  z = 0,1; \, d = 0,1\}}

\def\mr{\mathbb R}
\def\yi{Y_i}

\def\zi{Z_i}

\def\vi{V_i}
\def\ips{IV propensity score}

\def\xit{X_i^\T}
\def\vit{V_i^\T}
\def\ste{systematic treatment effect variation}

\def\hte{treatment effect heterogeneity} 
\def\xic{\check X_{i, \textup{dm}}}

\def\tc{\tau_\cc}
\def\dmid{(1(\xis = 1), \ldots, 1(\xis= K))^\T}

\def\htwk{\htau_{\text{wald},[k]}}

\def\tdvi{\widetilde{DV_i}}

\def\hdvi{\widehat{DV_i}}

\def\uc{U_i = \cc}
\def\cp{U_i = \cc}

\def\agv{Assume $X_i = (V_i^\T, W_i^\T)^\T \in \mathbb R^{K}$, where $V_i \in \mathbb{R}^Q$ and $W_i \in \mathbb{R}^{K-Q}$ with $Q \leq K$}

\def\epi{\epsilon_i}

\def\yizd{Y_i(z,d)}

\def\tdxi{ \widetilde{DX}_i}

\def\hdxi{\widehat{D X}_i}

\def\bc{\beta_\cc}
\newcommand{\bcd}{\beta_{\cc,\dd}}

\def\aa{\textup{a}}
\def\cc{\textup{c}}
\def\dd{\textup{d}}
\def\nn{\textup{n}}

\def\hbi{\hb_\tsls}

\def\bi{\beta_\tsls}

\def\hb{\hat\beta}

\def\cc{\textup{c}}

\def\zi{Z_i}

\def\tk{\tau_{[k]\cc}}

\def\hta{\htau_{++}}
\def\hti{\htii}
\def\htii{\htau_{\times\times}}

\def\htia{\htau_{\times+}}

\def\ta{\tau_{++}}
\def\ti{\tau_{\times\times}}

\def\tia{\tau_{\times+}}

\def\sumk{\sum_{k=1}^K}

\def\xxi{X_i}

\def\xis{X_i^*}

\def\hdi{\hat D_i}

\def\tdi{\tilde D_i}

\def\dizz{D_i(z)}
\def\yizzd{Y_i(z,d)}

\def\yiz{Y_i(0)}
\def\yio{Y_i(1)}
\def\diz{D_i(0)}
\def\dio{D_i(1)}
\def\hmu{\hat\mu}

\def\beginp{\begin{pmatrix}}
\def\endp{\end{pmatrix}}
\def\ols{\textsc{ols}}

 \newcommand{\muk}{\mu_k}

\def\ep{\epsilon}

\newcommand{\sm}{Supplementary Materials}

\newcommand{\di}{D_i}

\def\zi{Z_i}
\def\beginp{\begin{pmatrix}}
\def\endp{\end{pmatrix}}

\newcommand{\htau}{\hat\tau}

\def\begina{\begin{eqnarray*}}
\def\enda{\end{eqnarray*}}
\def\beginals{\begin{align*}}
\def\endals{\end{align*}}

\def\beginar{\begin{array}}
\def\endar{\end{array}}

\def\begineqs{\begin{equation*}}
\def\endeqs{\end{equation*}}
\def\begineq{\begin{equation}}
\def\endeq{\end{equation}}
\def\beginy{\begin{eqnarray}}
\def\endy{\end{eqnarray}}

\def\begini{\begin{itemize}}
\def\endi{\end{itemize}}
\def\begine{\begin{enumerate}}
\def\ende{\end{enumerate}}

\newcommand{\indep}{\perp \!\!\! \perp}

\DeclareMathOperator\res{Res} 
\DeclareMathOperator\proj{Proj} 
\DeclareMathOperator\E{\mathbb E}    
\DeclareMathOperator\var{var}

\newcommand{\T}{\top}

\newcommand{\wrt}{with respect to}

\renewcommand{\sec}{Section}
\newcommand{\secs}{Sections}
\def\prop{Proposition}
\newcommand{\thm}{Theorem}
\newcommand{\lem}{Lemma}

\def\assm{Assumption}
\def\assms{{\assm}s}

\def\nnb{\nonumber}

\def\zi{Z_i}
\def\yi{Y_i}
\def\xxi{X_i}

\def\yio{Y_i(1)}
\def\yiz{Y_i(0)}

\newcommand{\mbr}{\mathbb R}

\newcommand{\oeqt}[1]{\overset{\textup{#1}}{=}}
\newcommand{\oeq}[1]{\overset{#1}{=}}

\newcommand{\ot}[1]{1, \ldots, #1}

\newcommand{\pr}{\mathbb P}

\def\ep{\epsilon}

\begin{document}

\spacingset{1}

  \title{\bf \Large Two-stage least squares with treatment-covariate interactions for \hte}
  \date{}
  \author{Anqi Zhao\thanks{
  Anqi Zhao: Fuqua School of Business, Duke University. Peng Ding: Department of Statistics, University of California, Berkeley. Fan Li: Department of Statistical Science, Duke University. 
    Peng Ding was supported by the U.S. National Science Foundation (grants \#1945136, \#2514234).}
    \quad Peng Ding \quad 
    Fan Li 
    }

  \maketitle

\bigskip

\begin{abstract}
Treatment effect heterogeneity with respect to covariates is common in instrumental variable (IV) analyses. 
An intuitive approach, which we call the {\it interacted two-stage least squares} (\tsls),  is to regress the outcome on the treatment, covariates, and treatment-covariate interactions, 
using the IV, covariates, and IV-covariate interactions as instruments. 
We clarify the causal interpretation of the \its\ under the local average treatment effect (LATE) framework when the IV is valid conditional on the covariates.
Our main findings are twofold.  
%
First, we show that the coefficients on the treatment-covariate interactions from the \its\ identify \hte\ \wrt\ covariates among compliers for any outcome-generating process if and only if \condivf, 
a requirement we call the {\it linear IV-covariate interactions} condition.
Second, when the covariate vector has dimension $K$ and includes a constant term, we show that the linear IV-covariate interactions condition requires the IV propensity score to take at most $K$ distinct values. 
The condition holds when (a) the covariates take exactly $K$ distinct values, or (b) the IV is randomly assigned independently of the covariates, but is restrictive beyond these cases.
These results underscore the difficulty of interpreting regression coefficients from specifications with treatment-covariate interactions when the covariates are not saturated and the IV is not unconditionally randomized, absent correct specification of the outcome model.   
We then connect these results to the literature on LATE estimation, and show that under the linear IV-covariate interactions condition, the \its\ with covariates demeaned by their complier means recovers the LATE as a single coefficient. 
%
\end{abstract}

\noindent%

\vfill

\newpage
\spacingset{1.6} 

\section{Introduction}
Two-stage least squares (\tsls) is widely used to estimate causal effects of endogenous treatments using instrumental variables (IVs). 
Under the potential outcomes framework, \cite{imbens1994identification} and \cite{angrist1996identification} defined the local average treatment effect (LATE) as the average treatment effect among {\it compliers} whose treatment status is affected by the IV, and established conditions under which \tsls\ consistently estimates the LATE.
 
Often, IV validity holds only after conditioning on a set of covariates. 
We call such IVs {\it conditionally valid}. 
A common strategy is to include the corresponding covariates as additional regressors in \tsls.
When \hte\ \wrt\ covariates is suspected, an intuitive generalization, which we term the {\it interacted} \tsls, is to add IV-covariate and treatment-covariate interactions to the first and second stages, respectively \citep[Section 4.5.2]{angrist2009mostly}.

\cite{ding2019decomposing} analyze the \its\ when the IV is randomly assigned and thus {\it unconditionally valid}.
We extend this analysis to conditionally valid IVs and clarify the causal interpretation  of \its\ when the outcome working model may be misspecified. 
Our two main findings are as follows. 

First, define {\it \ste} as the variation in individual treatment effects that can be linearly explained by covariates \citep{heckman1997making, djebbari2008heterogeneous, ding2019decomposing}. We show that the coefficients on treatment-covariate interactions from the \its\ identify the {\ste} among compliers for any outcome-generating process 
if and only if \condivf, a requirement we call the {\it linear IV-covariate interactions} condition. 
%
This extends the necessity and sufficiency of the {\it rich-covariates} condition in \cite{blandhol2026when} for a causal interpretation of additive \tsls\ to specifications with treatment-covariate interactions.

Second, assuming that the covariate vector has dimension $K$ and includes a constant term, we show that the linear IV-covariate interactions condition holds only if the IV propensity score takes at most $K$ distinct values. 
The condition holds in two important cases:
(a) the covariates take exactly $K$ distinct values (a saturated specification), or (b) the IV is randomly assigned independently of the covariates.
Beyond these cases, the restriction can be demanding. 

Together, these findings show that conditional IV validity alone does not guarantee a causal interpretation of \its\ coefficients under outcome-model misspecification, particularly when the covariates are not saturated and the IV is not unconditionally randomized.
This extends the discussion of interacted \ols\ in \citet[Section 3.2]{chen2025potential} to the IV setting. 

We then connect these findings to existing approaches to LATE estimation. 
We show that under the linear IV-covariate interactions condition, the \its\ with covariates demeaned by their complier means identifies the LATE as the coefficient on treatment.
With saturated categorical covariates and moment-based estimates of the category shares among compliers, this coefficient is numerically identical to the nonparametric estimator of \cite{tan2006regression} and \cite{frolich2007nonparametric}. 
When the IV is independent of the covariates, by contrast, additive \tsls\ already identifies the LATE, so adding interactions provides no identification gain.

To further inform practice, we extend our theory to {\it partially \its} specifications that include interactions between treatment and only a subset of covariates, as in \cite{resnjanskij2024can}.
We show that omitting interactions changes both the target estimand and the conditions for its identification.
See also \cite{hirano2001estimation} and \cite{hainmueller2019much} for applications of partially interacted ordinary least squares (\ols) regressions.

\paragraph{Notation.} 
For a population indexed by $\otn$, 
let $u_i\in \mbr$, 
$v_{ij} \in \mbr^{p_j}$ for $j = \ot{J}$, and 
$w_{il} \in \mbr^{q_l}$ for $l = \ot{L}$.
Define $v_i = (v_{i1}^\T, \ldots, v_{iJ}^\T)^\T$ and $w_i = (w_{i1}^\T, \ldots, w_{iL}^\T)^\T$. 
Let $\lmt(u_i \sim v_{i1}+\cdots+v_{iJ})$ denote the \ols\ regression of $u_i$ on $v_i$ over $\otn$, and $\tslst(u_i \sim v_{i1}+\cdots+v_{iJ} \mid w_{i1}+\cdots+w_{iL})$ the \tsls\ regression of $u_i$ on $v_i$ over $\otn$, instrumented by $w_i$. 
Each $v_{ij}$ and $w_{il}$ may be a scalar or vector; + denotes concatenation of regressors. 
Throughout, we use $\lmt(\cdot)$ and $\tslst(\cdot)$ to denote the numerical outputs of \ols\ and \tsls\ without imposing any assumption about the corresponding linear model. 
For random vectors $Y \in \mathbb R^p$ and $X \in \mathbb R^q$, let
$\proj(Y\mid X)$ denote the linear projection of $Y$ onto $X$, given by $\proj(Y\mid X) = BX$, where $B = \textup{argmin}_{b \in \mathbb R^{p\times q}}\E(\|Y - bX\|^2) = \E(Y X^\T)\{\E(X X^\T)\}^{-1}$.
Let $1(\cdot)$ denote the indicator function. 
Let $\indep$ denote independence.

\section{Interacted 2{\large SLS} and identifying assumptions}\label{sec:setup}

\subsection{Interacted 2{\normalsize SLS}}
Consider a study with two treatment levels, indexed by $d = 0,1$, and a study population of $N$ units, indexed by $i = \ot{N}$. 
For each unit, we observe a binary treatment $\di \in \{0,1\}$, an outcome $Y_i \in \mathbb R$, a $K\times 1$ vector of exogenous baseline covariates $X_i$, and a binary IV $Z_i \in \{0,1\}$ that is valid conditional on $ X_i$.
We formalize conditional IV validity in \assm~\ref{assm:iv} in \sec~\ref{sec:assumptions}. 

Definition~\ref{def:2sls_a} below reviews the standard \tsls\ specification for estimating the causal effect of the treatment on the outcome. 
We call it the {\it additive \tsls} to signify the additive regression specifications in both the first and second stages.

\begin{definition}[Additive \tsls]\label{def:2sls_a}
Consider the \tsls\ regression of $Y_i$ on $(\di, \xxi)$, instrumented by $(\zi, \xxi)$ over $i = \ot{N}$: 
$$
\tslst(Y_i \sim \di + \xxi \mid \zi + X_i).
$$  
Estimate the treatment effect by the coefficient on $\di$.
\end{definition}

The additive \tsls\ in Definition~\ref{def:2sls_a} can be viewed as motivated by the {\it additive working model} $Y_i = D_i \beta_D + X_i ^\T \beta_X + \eta_i$, where $\E(\eta_i \mid Z_i, X_i) = 0$, so that $\eta_i$ is mean independent of $(\zi,\xxi)$. 
Under this working model, the coefficient $\beta_D$ represents a constant treatment effect, and we use $Z_i$ to instrument the possibly endogenous $D_i$. 
We use the term {\it working model} to refer to a model that motivates an estimator without requiring correct specification.

When \hte\ \wrt\ covariates is suspected, an intuitive generalization is to consider the {\it interacted} working model 
\begineq\label{eq:wm}
Y_i = D_i \xit \bdx + X_i ^\T \beta_X + \eta_i,
\endeq 
where $\E(\eta_i \mid Z_i, X_i) = 0$. 
The linear function $\xit \beta_{DX}$ models systematic treatment effect variation with respect to covariates \citep{heckman1997making, ding2019decomposing}. 
We use $\zxi$ to instrument the possibly endogenous $D_iX_i$; see, e.g., \citet[][Section 4.5.2]{angrist2009mostly}.
Definition~\ref{def:2sls_i} below formalizes the resulting {\it interacted \tsls}. 

\begin{definition}[Interacted \tsls]\label{def:2sls_i}
Consider the \tsls\ regression of $Y_i$ on $(\dxi, \xxi)$, instrumented by $(\zxi, \xxi)$ over $i = \ot{N}$: 
$$
\tslst(Y_i \sim \dxi + \xxi \mid \zxi + X_i).
$$   
Let $\hbi$ denote the coefficient vector of $\dxi $.
\end{definition}

Although the interacted \tsls\ is motivated by the working model in \eqref{eq:wm}, we study the causal interpretation of the coefficient vector $\hbi$ under the LATE framework, which views the \its\ as a numerical device for computing $\hbi$ without assuming that the working model \eqref{eq:wm} is correctly specified.

\subsection{IV assumptions and causal estimands}\label{sec:assumptions}
We formalize the IV assumptions and estimands using the potential outcomes framework \citep{imbens1994identification, angrist1996identification}.

For $z, d \in \{0,1\}$, 
let
$\dizz$ denote the potential treatment status of unit $i$ under $Z_i=z$, and let 
$\yizzd$ denote its potential outcome under $Z_i = z$ and $\di = d$. 
The observed treatment status and outcome satisfy $\di = \di(\zi)$ and $\yi = \yi(\zi,\di) = \yi(\zi,\di(\zi))$. 
The pair $(\dio, \diz)$ classifies the units into four compliance types. 
We call unit $i$ 
an {\it always-taker} if $\dio = \diz = 1$, denoted by $U_i = \aa$; 
a {\it complier} if $\dio = 1$ and $\diz = 0$, denoted by $U_i = \cc$; 
a {\it defier} if $\dio = 0$ and $\diz = 1$, denoted by $U_i = \dd$; and 
a {\it never-taker} if $\dio = \diz = 0$, denoted by $U_i = \nn$.

Throughout, assume that 
$\{\yizzd, \dizz, X_i, Z_i: z, d = 0,1\}$ are independent and identically distributed across $i = \ot{N}$, and let $\pr$ and $\E$ denote the probability and expectation with respect to the common joint distribution of the unit vector $\{\yizzd, \dizz, X_i, Z_i: z, d = 0,1\}$.
By the law of large numbers, the probability limit of $\hbi$ equals
\begineq\label{eq:bi_def}
\bi = \text{the first $K$ elements of $\left[ \E\left\{\beginp
\zxi\\
X_i
\endp \Big(D_iX_i^\T, X_i^\T\Big)\right\} \right]^{-1} \E\left\{\beginp
\zxi\\
X_i
\endp Y_i \right\}$},
\endeq 
provided the standard rank condition that the matrix inverse in \eqref{eq:bi_def} is well defined.
Of interest is the causal interpretation of $\bi$ under the LATE framework. 

As a basis, 
Assumption~\ref{assm:iv} below reviews the assumptions for {\it conditionally valid} IVs underlying our main results.

\begin{assumption}\label{assm:iv}
\begine[(i)]
\item\label{it:indep} \textit{Conditional exogeneity}: $Z_i \indep \ydset \mid X_i$. 
\item\label{it:er} \textit{Exclusion restriction}: $Y_i(0, d) = Y_i(1,d) = \yid$  for $d = 0,1$, where $\yid$ denotes the common value.  
 \item\label{it:overlap} \textit{Overlap}: $0 < \pr(Z_i = 1\mid \xxi) < 1$.    
\item\label{it:relevance} \textit{Relevance}: $\E\{\dio - \diz \mid X_i\}\neq 0$.
\item\label{it:mono} \textit{Monotonicity:} $\dio \geq \diz$.   
\ende 
 \end{assumption}

\assmiv\eqref{it:indep} requires that $Z_i$ is as-if randomly assigned given $ X_i$.
\assmiv\eqref{it:er} rules out direct effects of \(Z_i\) on the outcome.
\assmiv\eqref{it:overlap} ensures that both IV levels have positive probability for all values of $X_i$.
\assmiv\eqref{it:relevance} ensures the IV has a nonzero average causal effect on treatment status at all possible values of $X_i$, so that $\pr(\uicd\mid \xxi) > 0$.
\assmiv\eqref{it:mono} precludes defiers. 
Together, \assmiv\eqref{it:relevance}--\eqref{it:mono} imply a positive complier probability at every covariate value, i.e., $\pr(\cp \mid X_i) > 0$.
\assmiv\eqref{it:mono} is also known as {\it strong monotonicity} \citep{sloczynski2022not}.
\cite{sloczynski2022not} and \cite{blandhol2026when} studied a weaker version of it, known as {\it \wmf}, which requires only that, at each covariate value, either compliers or defiers be absent.
We maintain \assmiv\ throughout most of the main paper, and extend our results to \wmf\ in  \sec~\ref{sec:ext}.

Under \assmiv, define 
$\tau_i = \yio - \yiz $ as the individual treatment effect of unit $i$, where $\yid$ is the common value defined in Assumption~\ref{assm:iv}\eqref{it:er}. 
Define
\begineq\label{eq:tc}
 \tc = \E(\tau_i \mid \cp )
\endeq
as the LATE, and define 
\begineq\label{eq:tcx}
 \tc( X_i ) = \E(\tau_i \mid X_i,\, \cp)
\endeq
as the conditional LATE given $X_i$. 
The law of iterated expectations implies that $\tc =\E\{\tc(X_i) \mid \cp \}$. 

To quantify \hte\ \wrt\ covariates, 
consider the linear projection of $\tau_i$ onto $X_i$ among compliers, denoted by
$\proj_{\,\uc}(\tau_i \mid X_i) = \xit \bc$. 
The coefficient vector is 
\begineq\label{eq:bc_def}
\beginar{rcl}
\bc &=& \textup{argmin}_{b \in \mathbb R^K}\E\left\{  (\tau_i - \xit b )^2 \mid \uc\right\}\\
&=& \left\{\E( X_i \xit \mid \cp ) \right\}^{-1}\E\left( X_i \tau_i \mid \cp \right).
\endar
\endeq
By \citet[Theorem 3.1.6]{angrist2009mostly}, $\xit \bc $ is also the linear projection of the conditional LATE $\tc( X_i)$ on $ X_i$ among compliers, with 
\begineqs
\bc = \textup{argmin}_{b \in \mathbb R^K}\E[ \{\tc(X_i) - \xit b\}^2 \mid \uc].
\endeqs
Therefore, $\xit \bc$ is the best linear approximation to both $\tau_i$ and $\tc(X_i)$ based on $X_i$ among compliers, generalizing the notion of {\it systematic treatment effect variation} in \cite{djebbari2008heterogeneous} to the IV setting; see also \cite{heckman1997making} and \cite{ding2019decomposing}. 
This justifies viewing $\bc$ as the causal estimand for quantifying the part of \hte\ among compliers that is linearly explained by $X_i$.

In a finite-population, design-based framework, \cite{ding2019decomposing} showed that $\bi=\bc$ when $Z_i$ is randomly assigned and thus valid without conditioning on $X_i$. 
We establish in Sections~\ref{sec:hte}--\ref{sec:late} the properties of $\bi$ for identifying $\bc$, $\tc$, and $\tc(X_i)$ when $Z_i$ is only conditionally valid given $X_i$. 

\begin{remark}
Assumption~\ref{assm:iv}\eqref{it:relevance} is stronger than necessary to define the population-level complier estimands $\tc$ and $\bc$.
The LATE $\tc$ remains well defined if compliers are absent at some covariate values, provided that the overall complier proportion $\pr(\uic)>0$, while $\bc$ additionally requires that $\E(X_iX_i^\top\mid \uc)$ is invertible.
Covariate values at which $\pr(\uic\mid \xxi)=0$ receive zero weight in the complier distribution and therefore do not affect either estimand. 
They do, however, leave the conditional LATE $\tc(X_i)$ undefined. 
In particular, with saturated categorical covariates, a category with no compliers has an unidentified category-specific coefficient.
We impose pointwise relevance to avoid these support-specific qualifications.
\end{remark}

\section{Causal interpretation of $\bi$}\label{sec:hte} 
We maintain \assmiv\ throughout this section. 
When the interacted working outcome model in \eqref{eq:wm} is correctly specified, 
standard \tsls\ theory implies that $\bi$ equals the coefficient on $\dxi$, $\bdx$. 
The conditional Wald identity then gives $\tc(\xxi) = \xit\bdx$, so the conditional LATE $\tcx$ is linear in $\xxi$, with $\bc = \bdx$.

Beyond this benchmark case, we present in \sec~\ref{sec:bi} a necessary and sufficient condition on the IV assignment mechanism for $\bi$ to identify $\bc$ across all potential-outcome models, without requiring that the working model in \eqref{eq:wm} is correctly specified. 
We then study the stringency of this condition in Section~\ref{sec:nece}.

\subsection{A necessary and sufficient condition}\label{sec:bi}

As a starting point, \defli\ below reviews the {\it level independence} condition due to \cite{blandhol2026when} in the current context under \assmiv. See also the {\it minimally quasi-experimental} condition in \cite{chen2025potential}. 
Let $\mathcal P$ denote the collection of all joint distributions of $\{\yid, \dizz, \xxi, Z_i: z,d=0,1\}$ satisfying \assmiv.

\begin{definition}[Level independence] \label{def:li}
Under \assmiv, a functional defined on $\mathcal P$ is {\it level-independent} if, for any $\pr\in\mathcal P$, it depends on $\{\yio, \yiz\}$ only via $\tau_i = \yio - \yiz$.  
\end{definition}

Level independence is arguably a minimal requirement for a functional on $\mathcal P$ to be interpretable as an average treatment effect: such an effect averages within-unit comparisons of $\yio$ and $\yiz$, and must therefore depend on them only through $\tau_i$.
The LATE $\tc$ in \eqref{eq:tc}, the conditional LATE $\tc(\xxi)$ in \eqref{eq:tcx}, and the projection coefficient $\bc$ in \eqref{eq:bc_def} are all level-independent functionals on $\mathcal P$. 

Let 
\begineqs
e(x) = \pr(\zi = 1\mid \xxi = x)
\endeqs 
denote the {\it IV propensity score}  \citep{rosenbaum1983central}.
Assumption~\ref{assm:ivps} below states a condition on $\exi$ that is necessary and sufficient both for $\bi$ to be level-independent and for $\bi = \bc$ across all potential-outcome models.

\begin{assumption}[\bfseries Linear IV-covariate interactions]\label{assm:ivps} $\E(\zxi \mid X_i) = e(X_i)X_i$ is linear in $X_i$. 
\end{assumption}

\begin{theorem}\label{thm:bi}
Assume \assmiv. Then 
\begine[(i)]
\item\label{it:thm_bi_nece} $\bi$ is level-independent if and only if \assmivps\  holds. 
\item\label{it:thm_bi_suff} $\bi = \bc$ for all possible potential-outcome models if and only if \assmivps\ holds. 
\ende
\end{theorem}

Given that $\bc$ is level-independent, the necessity of \assmivps\ in \thmbi\eqref{it:thm_bi_nece} implies its necessity in \thmbi\eqref{it:thm_bi_suff}, while 
the sufficiency in \thmbi\eqref{it:thm_bi_suff} implies its sufficiency in \thmbi\eqref{it:thm_bi_nece}. 

\cite{blandhol2026when} showed that the {\it rich-covariates} condition, 
\begineq\label{eq:rich}
\exi = \proj(\zi\mid \xxi), 
\endeq
is necessary and sufficient for additive \tsls\ to admit a causal interpretation for arbitrary potential-outcome models.
\thmbi\ extends this characterization to specifications with treatment-covariate interactions.

When $\xxi$ includes a constant term, or there exists a constant vector $c\in\mbr^K$ such that $\xit c$ is a nonzero constant, \assmivps\ implies the rich-covariates condition in \eqref{eq:rich}. 
Write $\exi = \proj(\zi\mid \xxi) = \xit\gamma$, 
where $\gamma = \{\E(\xxi \xit)\}^{-1}\E(\xxi\zi)$. 
Under this representation, \assmivps\ is equivalent to requiring that $X_i \xit \gamma$ be linear in $\xxi$.
Because $X_i \xit\gamma$ involves products of the nonconstant components of $X_i$, this is restrictive, as we formalize below. 

\subsection{Sufficient cases and necessary restrictions for \assmivps}\label{sec:nece}
\thm~\ref{thm:nece} below gives two sufficient cases for \assmivps\ and a necessary restriction that substantially limits the admissible IV assignment mechanisms.   

\begin{theorem}\label{thm:nece}
Let $K$ denote the dimension of the covariate vector $X_i$.
\begine[(i)]
 \item\label{it:thm_nece_sufficient} \assmivps\ holds if either of the following conditions holds:
\begine 
\item\label{it:thm_cat} \textbf{Categorical covariates:} $X_i$ takes exactly $K$ distinct values. 
\item\label{it:thm_indep} \textbf{Random assignment of IV:} $Z_i \indep X_i$. 
\ende 
\item \label{it:assm_ezx_necessary} 
Assume that $X_i = (1, \xio , \ldots, X_{i,K-1})^\T$ includes the constant term as the first element. 
If \assmivps\ holds, then $\exi$ must be linear in $X_i$ and take at most $K$ distinct values over the support of $X_i$.

Let $e_a(x) = a^\T x$ for $a \in \mbr^K$. 
If \assmivps\ holds for all $\{e_a(x):a \in \mbr^K\}$, then $\xxi$ must take exactly $K$ distinct values.
\end{enumerate}
\end{theorem}

\thm~\ref{thm:nece}\eqref{it:thm_nece_sufficient} provides two sufficient conditions for Assumption \ref{assm:ivps}. 
These are also sufficient cases for the rich-covariates condition discussed in \cite{blandhol2026when}. 
\begini
\item The categorical-covariate condition in \thm~\ref{thm:nece}\eqref{it:thm_cat} is equivalent to the {\it saturated covariates} condition in \cite{angrist2009mostly}, where $X_i$ encodes a $K$-level categorical variable $\xis \in \{\ot{K}\}$ using category indicators, $X_i = \dmid$. 
As a standard saturated-covariate result, the \its\ is numerically equivalent both to $K$ category-specific \tsls\ regressions and to conditioning nonparametrically on the covariate categories.
Hence, $\bi = \bc = (\tau_{[1]\cc}, \ldots, \tau_{[K]\cc})^\T$, where $\tk = \E(\tau_i \mid X_i^* =k, \uc)$ is the conditional LATE for category $k$, and $\tc(X_i) = \xit\bc$. 
By the invariance of \tsls\ to invertible linear transformations of the covariates, $\bi=\bc$ and $\tc(X_i)=X_i^\top\bc$ continue to hold under any full-rank saturated encoding of the categories, although $\bc$ no longer equals the vector of category-specific LATEs. 
For nonsaturated $X_i$, $X_i^\top\bc$ is the linear projection of $\tc(X_i)$ onto $X_i$ among compliers and need not equal $\tc(\xxi)$.

\item Combined with Assumption \ref{assm:iv}, the condition $Z_i \indep X_i$ in \thm~\ref{thm:nece}\eqref{it:thm_indep}  implies that the IV is unconditionally valid. By \thmbi, the resulting $\bi$ identifies $\bc$.  See \citet[][Theorem 7]{ding2019decomposing} for similar results in the finite-population, design-based framework.
\endi

More importantly, \thm~\ref{thm:nece}\eqref{it:assm_ezx_necessary} shows that when $\xxi$ includes a constant term, 
\assmivps\ requires not only the rich-covariates condition in \eqref{eq:rich}, but also that $\exi$ takes at most $K$ distinct values. 
This substantially restricts the admissible IV assignment mechanisms beyond the two special cases in \thm~\ref{thm:nece}\eqref{it:thm_nece_sufficient}.
By an invertible linear transformation, the same restriction holds whenever there exists a constant vector $c\in\mbr^K$ such that $\xit c$ is a nonzero constant. 
This requirement rules out the possibility left open under the rich-covariates condition that $\exi$ is a linear function of $\xxi$ taking more than $K$ distinct values.

In addition, \thm~\ref{thm:nece}\eqref{it:assm_ezx_necessary} implies that if \assmivps\ holds for every linear function $e_a(x) = a^\T x$, 
meaning that $e_a(X_i)X_i$ is linear in $X_i$ for every $a \in \mbr^K$, 
then $\xxi$ must take exactly $K$ distinct values.
We sketch the argument below to build intuition.

\begin{proof}[\bf Proof sketch of \thm~\ref{thm:nece}\eqref{it:assm_ezx_necessary}]
Let $\txi = \txif$, so that $X_i = (1, \txi^\T)^\T$. 
\assmivps\ is equivalent to
\beginy\label{eq:linear_main}
\text{$e(\xxi) = \E(Z_i \mid X_i)$ and $\E(Z_i \txi \mid X_i)$ are both linear in $X_i$}.  
\endy
From \eqref{eq:linear_main}, there exist constants $a_0\in \mr$ and $a_X\in\mr^{K-1}$ such that 
\begineq\label{eq:ezx_main}
e(X_i) = \E(Z_i \mid X_i) = a_0 + \txi^\T a_X,
\endeq
and 
\begineq\label{eq:ezx_1:K-1_main}
\E(Z_i \txi \mid X_i) =\txi \cdot \E(\zi \mid X_i)  \overset{\eqref{eq:ezx_main}}{=} 
 \txi a_0 + \txi \txi ^\T a_X.
\endeq
For \eqref{eq:ezx_1:K-1_main} to be linear in $(1, \txi )$ as suggested by \eqref{eq:linear_main}, we need $ \txi \txi ^\T a_X$ to be linear in $ (1, \txi ) $.
In \lem~\ref{lem:la_x} in the \sm, we show that this implies $\txi ^\T a_X$ takes at most $K$ distinct values. By \eqref{eq:ezx_main}, so does $e(X_i)$. 
As an illustration, consider the case of $K=2$ with $X_i = (1, X_{i1})^\T$, where $\txi = X_{i1}\in \mr$ is a scalar.
Equation~\eqref{eq:ezx_1:K-1_main} implies that $
\E(Z_i \txi \mid X_i) = \txi a_0 + \txi^2 a_X$. 
\assmivps\ requires this expression to be linear in $\xxi$, so
\begineqs
 \txi a_0 + \txi^2 a_X = b_0 + b_X \txi  
\endeqs
for some constants $b_0$ and $b_X$. 
If $a_X \neq  0$, this quadratic equation has at most two distinct solutions for $\txi$.
If $a_X = 0$, $e(X_i) = a_0$ is constant. 
Thus, $e(X_i)$ takes at most two distinct values in either case.
The general proof uses an eigenvalue argument; see \lem~\ref{lem:la_x} in the \sm.

In addition, if $e_a(\xxi)\xxi$ is linear in $X_i$ for every $a \in \mbr^K$, then from \eqref{eq:ezx_1:K-1_main}, we need $ \txi \txi ^\T a_X$ to be linear in $ (1, \txi )$ for every $a_X$. 
This implies that 
\begineq\label{eq:intuition}
\text{$\txi\txi^\T$ is componentwise linear in $(1,\txi)$.}
\endeq 
We show in \lem~\ref{lem:la_x} in the \sm\ that \eqref{eq:intuition} restricts $\txi$, and therefore $\xxi$, to at most $K$ distinct values. 
Recall the rank condition for $\bi$ in \eqref{eq:bi_def}, which requires that $\xxi$ takes at least $K$ distinct values. 
Therefore, $X_i$ takes exactly $K$ distinct values.  
\end{proof}

Together, Theorems~\ref{thm:bi}--\ref{thm:nece} extend the discussion in \cite{chen2025potential} of the challenges facing specifications with treatment-covariate interactions from the estimation of average treatment effects under ordinary least squares to the estimation of \hte\ under \tsls. 

\section{Unification with the literature on LATE estimation}\label{sec:late}
\subsection{General identification results}
We now compare existing \tsls\ and nonparametric approaches to LATE estimation under \hte.
\citet{angrist1995two}, \cite{kolesar2013estimation}, and \cite{sloczynski2022not} show that \tsls\ specifications with an additive second stage generally identify covariate-weighted combinations of conditional LATEs that differ from $\tc$.
For the additive \tsls\ in Definition~\ref{def:2sls_a}, these weights are positive under the rich-covariates condition in \eqref{eq:rich}, which requires $\exi$ to be linear in $X_i$. 
Without this condition, \citet{blandhol2026when} show that the weights can be negative, so the additive \tsls\ need not admit a causal interpretation. 

Separately, \cite{tan2006regression} and \cite{frolich2007nonparametric} proposed nonparametric estimators that identify $\tc$ by aggregating covariate-specific reduced forms and first stages. 
We connect these approaches to regression through the \its\ with covariates demeaned by their complier means---the {\it complier-demeaned \its} hereafter.
Under \assmivps, the complier-demeaned \its\ identifies the LATE as the coefficient on treatment.   
With saturated categorical covariates and moment-based estimates of the category shares among compliers, it provides a single-coefficient regression representation of the estimator common to \cite{tan2006regression} and \cite{frolich2007nonparametric}.
We formalize these connections below. 
Throughout this section, assume that $X_i = (1, \xio , \ldots, X_{i,K-1})^\T$ includes a constant, and let $\txi  = (\xio , \ldots, X_{i,K-1})^\T$ denote its nonconstant components.

First, Definition~\ref{def:2sls_hybrid} reviews the interacted-additive \tsls\ procedure studied by \cite{angrist1995two} under saturated covariates and later considered by \cite{kolesar2013estimation} and \citet{sloczynski2022not}.

\begin{definition}[Interacted-additive \tsls]\label{def:2sls_hybrid}
Consider the \tsls\ regression of $Y_i$ on $(\di, \xxi)$, instrumented by $(\zxi, \xxi)$ over $i = \ot{N}$: 
$$\tslst(Y_i \sim \di + \xxi \mid \zxi + X_i).$$
Let $\htia$ denote the coefficient on $\di$, and $\tia$ its probability limit. The subscript $\times+$ denotes the interacted first stage of $D_i$ on $(\zxi,\xxi)$ and the additive second stage of $\yi$ on $(\di, \xxi)$. 
\end{definition}

Definition~\ref{def:2sls_demeaned} below formalizes the complier-demeaned \its.
Let $\mu_k = \E(X_{ik} \mid \cp )$ denote the mean of $X_{ik}$ among compliers for $k = \ot{K-1}$. 
Because these means are generally unknown, 
demeaning requires first-step estimates.
Standard choices include the method-of-moments estimator of \citet[Equations (8)--(9)]{imbens1997estimating} and the $\kappa$-weighted estimator of \citet[Theorem 3.1]{abadie2003semiparametric}.
A componentwise implementation of the former may involve separate first-step calculations for each covariate. 
By contrast, $\kappa$-weighting uses the same weights for all covariates and can be more convenient when $K$ is large.
Let $\hmu_k$ be a consistent estimate of $\mu_k$ for $k = \ot{K-1}$.

\begin{definition}[Complier-demeaned \its] \label{def:2sls_demeaned}
Define the complier-demeaned covariate vector
$\xidm  = (1, \xio  - \hmu_1, \ldots, X_{i,K-1} - \hmu_{K-1})^\T = (1, \txidm^\T)^\T$, where $\txidm = (\xio  - \hmu_1, \ldots, X_{i,K-1} - \hmu_{K-1})^\T$. 
Replacing $\xxi$ with $\xidm$ in the \its\ in Definition~\ref{def:2sls_i} yields
\begineqs
\tslst(Y_i \sim  1 +D_i+  \txidm + D_i \txidm    \mid  1 +Z_i +  \txidm + Z_i \txidm   ).
\endeqs 
Let $\htii$ denote the coefficient on $\di$, and $\ti$ its probability limit. 
The subscript $\times\times$ denotes the interacted first and second stages. 
\end{definition}

Definitions~\ref{def:2sls_a}, \ref{def:2sls_hybrid}, and \ref{def:2sls_demeaned} specify three \tsls\ procedures: additive first and second stages, an interacted first stage with an additive second stage, and interacted first and second stages, respectively.
The remaining specification, an additive first stage with an interacted second stage, is underidentified when $K > 1$ and therefore omitted. 
Let $\hta$ denote the coefficient on $\di$ from the additive \tsls\ in Definition~\ref{def:2sls_a}, and let $\ta$ denote its probability limit. 
We compare the identification properties of $\taa$, $\tia$, and $\ti$ below. 
Let 
\begineqs
w(X_i)= \dfrac{\var\big\{\E(D_i\mid \zi , X_i)\mid \xxi\big\}}{ \E\big[\var\{\E(D_i\mid \zi , X_i) \mid \xxi\} \big]}, 
\endeqs
with $w(\xxi) > 0$ and $\E\{w(X_i)\} = 1$. 
Let $\pi(X_i) = \pr(U_i = \cc\mid X_i)$ denote the complier proportion given $X_i$, 
and let $\tilde\pi(X_i) = a^\T X_i$ denote the linear projection of $\pi(X_i)$ onto $X_i$ weighted by $\var(Z_i \mid X_i)$, where $a = [\E\{\var(Z_i \mid X_i) \cdot X_iX_i^\T\}]^{-1}\E\{ \var(Z_i \mid X_i) \cdot X_i \, \pi(X_i)   \}$.

\begin{proposition}\label{prop:unification}
Assume Assumption \ref{assm:iv}.
\begine[(i)]
\item\label{it:unification_taa} \citep[][Corollary~3.4]{sloczynski2022not} If $\exi$ is linear in $X_i$, then $\ta = \E\{ w_+(\xxi) \cdot \tc(\xxi) \}$, where
\begineqs
w_+(X_i) = \dfrac{\var(Z_i \mid X_i) \cdot \pi(X_i)}{\E\left\{\var(Z_i\mid X_i) \cdot \pi(X_i)\right\}},
\endeqs
with $w_+(X_i) > 0$ and $\E\{w_+(X_i)\} = 1$.
If $\E(D_i\mid \zi , X_i)$ is linear in $(\zi, X_i)$, so that the additive first stage in Definition \ref{def:2sls_a} is correctly specified, then $w_+( X_i) = w( X_i)$. 
\item\label{it:unification_tia} \citep[cf.][Theorem~3]{angrist1995two} If \assmivps\ holds, 
then $\tia = \E\{ w_\times(\xxi) \cdot \tc(\xxi) \}$, where
\begineqs
w_\times(X_i) = \dfrac{   \var(Z_i \mid X_i)\cdot \tilde\pi(X_i)   \cdot\pi(X_i)}{  \E\left\{ \var(Z_i\mid X_i) \cdot \tilde\pi(X_i)^2  \right\}},
\endeqs
with $\E\{w_\times(X_i)\} = 1$.
If $\E(D_i\mid \zi , X_i)$ is linear in $(\zi X_i, X_i)$, so that the interacted first stage in Definition \ref{def:2sls_hybrid} is correctly specified, then $w_\times( X_i) = w( X_i)$. 
\item\label{it:unification_tii} If \assmivps\ holds, then $\ti = \tc$.
\ende
\end{proposition}

Proposition~\ref{prop:unification}\eqref{it:unification_taa} reviews \citet[][Corollary 3.4]{sloczynski2022not}.
Because the LATE weights conditional LATEs by the complier shares, $\tc = \E\{\pi(\xxi)\tcx\}/\E\{\pi(\xxi)\}$, whereas $\ta$ uses the positive weights $w_+(\xxi)$, additive \tsls\
generally does not identify $\tc$ even when $\exi$ is linear in $\xxi$.
We further show that $w_+(X_i)$ reduces to $w(X_i)$ when the additive first stage is correctly specified. 
See \cite{sloczynski2022not} for a detailed discussion of the deviation of $\ta$ from $\tc$. 

Proposition~\ref{prop:unification}\eqref{it:unification_tia} gives the corresponding result for the interacted-additive \tsls. 
Although $w_\times(X_i)$ integrates to one, it may be negative where $\tilde\pi(X_i)$ is negative, so $\tia$ can be a signed weighted combination of the conditional LATEs.
With saturated covariates, 
\assmivps\ holds and $\E(D_i\mid Z_i, X_i)$ is linear in $(\zxi, X_i)$, so $\tilde\pi(X_i) = \pi(X_i)$, and Proposition~\ref{prop:unification}\eqref{it:unification_tia} reduces to \citet[][Theorem~3]{angrist1995two}.

Proposition~\ref{prop:unification}\eqref{it:unification_tii} establishes that, 
under \assmivps, 
the complier-demeaned \its\ consistently estimates the LATE for arbitrary potential-outcome models, paralleling the identification result for $\bc$ in \thm~\ref{thm:bi}\eqref{it:thm_bi_suff}.
Recall from \thm~\ref{thm:nece}\eqref{it:thm_nece_sufficient} that \assm~\ref{assm:ivps} holds when $X_i$ \ck\ or the IV is independent of the covariates. 
Under either scenario, $\hti$ estimates the LATE using a single \tsls\ coefficient with demeaned covariates. 
In particular, when $Z_i$ is independent of $X_i$, the weights $w_+(X_i)$ in \prop~\ref{prop:unification}\eqref{it:unification_taa} reduce to $\pi(X_i)/\E\{\pi(X_i)\}$, 
so additive \tsls\ already identifies $\tc$, 
and adding interactions yields no identification gain.
With saturated categorical covariates, $\hti$ coincides with the nonparametric estimator common to \cite{tan2006regression} and \cite{frolich2007nonparametric} for a specific choice of $\hmuk$, as shown below.

\subsection{A special case with saturated categorical covariates}\label{sec:tan}
Assume that $X_i$ encodes a $K$-level categorical variable $\xis \in \{\ot{K}\}$ using a constant and $K-1$ category indicators: $X_i = (1, 1(\xis = 1), \ldots, 1(\xis = K-1))^\T$. 
Let $\hat\Delta_{Y,k} = \hat\E(Y_i \mid X_i^* = k, Z_i = 1) - \hat\E(Y_i \mid X_i^* = k, Z_i = 0)$ and $\hat\Delta_{D,k} = \hat\E(D_i \mid X_i^* = k, Z_i = 1) - \hat\E(D_i \mid X_i^* = k, Z_i = 0)$
denote the estimated IV effects on the outcome and on the treatment within category $k$, where $\hat\E$ denotes the corresponding sample average.
Then 
\begineqs
\hti = \sumk \hmu_k \cdot \dfrac{\hat\Delta_{Y,k}}{\hat\Delta_{D,k}}, \qquad
\htau_\textup{nonpara} = \dfrac{\sumk \hat \pr(\xis =k) \hat\Delta_{Y,k}}{
\sumk \hat \pr(\xis =k) \hat\Delta_{D,k}},
\endeqs
where $\htau_\textup{nonpara}$ denotes the estimator common to \citet[][Equation~(9)]{tan2006regression} and \citet[][Theorem~1]{frolich2007nonparametric},  $\hat \pr(\xis = k)$ is the empirical category proportion, and $\hmu_k$ is the estimate of $\muk$ used for demeaning, with $\hmu_K = 1 - \sum_{k=1}^{K-1} \hmu_k$.
Thus, $\hti$ is a weighted average of the within-category Wald ratios, with weights given by the estimated category shares among compliers, whereas $\htau_\textup{nonpara}$ first aggregates the estimated IV effects on the outcome and on the treatment using the estimated marginal category
proportions and then takes their ratio. 
The two estimators coincide whenever the relevant ratios are well-defined and $\muk$ is estimated by 
$$
\hmuk 
= \dfrac{\hat\pr(\xis = k) \hat\Delta_{D,k}}{\sum_{j=1}^K \hat \pr(\xis = j) \hat \Delta_{D,j}},
$$
where, under \assm~\ref{assm:iv}, $\hat \Delta_{D,k}$ estimates the conditional complier proportion $\pr(\uic \mid \xis = k)$, so by Bayes' rule, $\hmuk$ gives a consistent estimate of $\muk = \pr(\xis = k \mid \uic)$. 

This choice of $\hmuk$ is the moment estimator of \cite{imbens1997estimating} applied to the categorical indicators, and is algebraically identical to the normalized $\kappa$-weighted estimator of \cite{abadie2003semiparametric} with the within-category empirical \ips. 
Thus, the complier-demeaned \its\ provides a single-coefficient regression representation of the estimator of \cite{tan2006regression} and
\cite{frolich2007nonparametric}.
With other consistent estimates of $\muk$, the two estimators remain consistent for $\tc$ but need not be numerically identical.
See \cite{hirano2001estimation} and \cite{lin2013agnostic} for similar uses of demeaned covariates when using interacted \ols\ regressions to estimate the average treatment effect. 

\section{Simulation}\label{sec:simu}
\subsection{Possible inconsistency for estimating $\bc$}\label{sec:simu_inconsistency}
We illustrate the inconsistency of the interacted \tsls\ for estimating $\bc$ when \assmivps\ fails. 
Consider the following data-generating process for  $(X_i, Z_i, D_i, Y_i)$:
\begine[(i)]
\item $X_i = (1, \xio )^\T$, where $\xio  \sim$ Uniform(0,1).
\item $Z_i \mid X_i \sim$ Bernoulli($e_i$), with $e_i = \xio $. 
\item $D_i = Z_i \cdot 1(U_i = \cc) + 1(U_i = \aa)$, where $U_i$ is independent of $(X_i, Z_i)$, with $\pr(U_i = \cc) = 0.7$, $\pr(U_i = \aa) = 0.2$, and $\pr(U_i = \nn) = 0.1$. 
\item $Y_i = D_i \yio + (1-D_i) \yiz$ with $Y_i(0) = 0$ and $\yio = \xio ^2$.
\ende
This data-generating process satisfies \assmiv\ but violates \assmivps, because \(e(X_i)X_i=(X_{i1},X_{i1}^2)^\top\) is not linear in \(X_i\). 
The data-generating process ensures that $\tau_i = \yio - \yiz = \xio ^2$ and 
$\bc = (-1/6, 1)^\T$. 

For each replication, we generate $N = 1,000$ independent realizations of $(X_i, Z_i, D_i, Y_i)$ and compute $\hbi$ from the interacted \tsls\ in Definition~\ref{def:2sls_i}.
Figure~\ref{fig:bias} shows the distributions of the two components of $\hbi-\bc$ over 1,000 replications, both concentrated away from zero.  

\begin{figure}[!t]
\begin{center}
\begin{tabular}{cc}
\includegraphics[width = .4\textwidth]{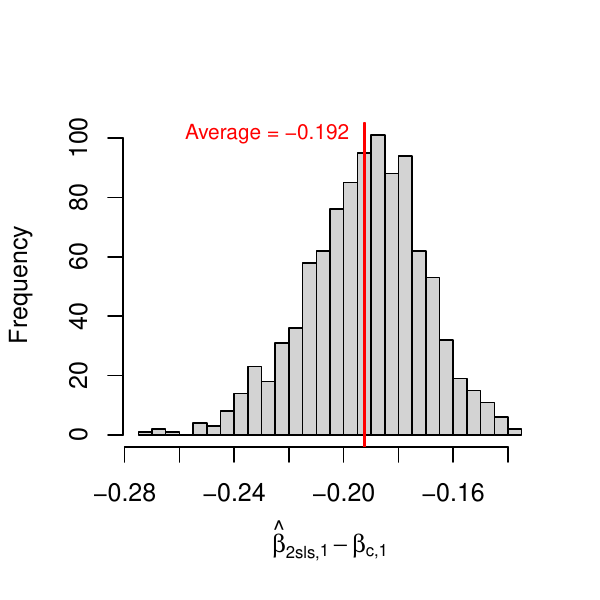}&\includegraphics[width = .4\textwidth]{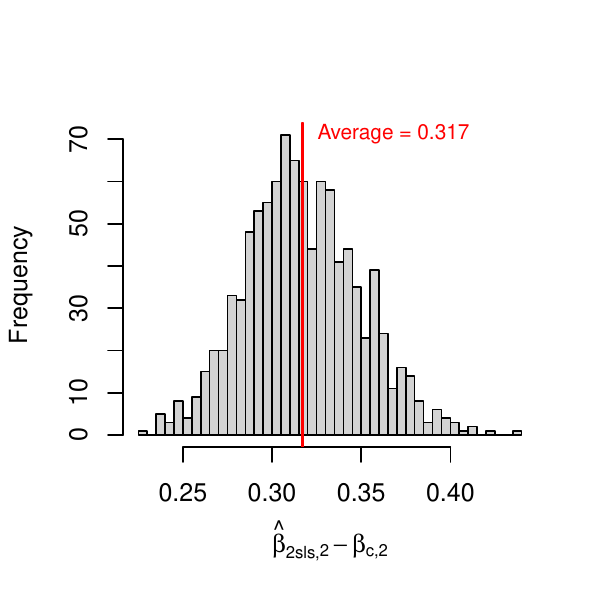}
\end{tabular}
\end{center}
\caption{\label{fig:bias} Distributions of $\hbi-\bc$ over 1,000 replications. The labels $\hat\beta_{\textsc{2sls},1}- \beta_{\cc,1}$ and $ \hat\beta_{\textsc{2sls},2}-\beta_{\cc,2}$ denote the first and second elements of $\hbi -\bc$, respectively.}
\end{figure}

\subsection{Complier-demeaned interacted 2{\normalsize SLS} for estimating the LATE}
We now illustrate \prop~\ref{prop:unification}\eqref{it:unification_tii} with a covariate-dependent IV propensity score satisfying \assmivps. 
Consider the following data-generating process for $(X_i, Z_i, D_i, Y_i)$:
\begine[(i)]
\item $X_i = (1, \xio )^\T$, where $\xio \sim$ Bernoulli(0.5).
\item  
$Z_i \mid X_i \sim$ Bernoulli$(0.5 + 0.4\xio)$. 
\item  $D_i = 1(U_i = \aa) + Z_i \cdot 1(U_i = \cc)$, where 
$U_i \indep Z_i \mid X_i$, with
$\pr(U_i = \aa \mid X_i) = 0.1$, $\pr(U_i = \cc \mid X_i) = 0.7 - 0.5\xio$, and $\pr(U_i = \nn \mid X_i) = 0.2 + 0.5\xio$. 
\item $Y_i = D_i \yio + (1-D_i) \yiz$, with $Y_i(0) = 0$ and $Y_i(1) = -1 + 5\xio$.
\ende
Because \(X_i\) takes exactly \(K=2\) distinct values, \assmivps\ holds by \thm~\ref{thm:nece}\eqref{it:thm_cat}, even though \(Z_i\) is not independent of \(X_i\).
The data-generating process ensures $\E(\tau_i \mid \xio = 1, \uc) = 4$, $\E(\tau_i \mid \xio = 0, \uc) = -1$, and $\tc =  1/9$, since \(\mathbb P(X_{i1}=1\mid U_i=c)=2/9\). 

For each replication, we generate $N = 10,000$ independent realizations of $(X_i, Z_i, D_i, Y_i)$ and compute $\hta$, $\htia$, and $\hti$ from the additive, interacted-additive, and complier-demeaned \its\ in Section~\ref{sec:late}; cf.\ Definitions~\ref{def:2sls_a}, \ref{def:2sls_hybrid}, and \ref{def:2sls_demeaned}. We use the method of moments to estimate the complier average of $\xio$ in computing $\hti$. 

Figure~\ref{fig:LATE} shows the distributions of the three estimators over 1,000 replications. 
Although the true $\tc$ is positive, 
almost all realizations of $\hta$ and $\htia$ fall below $-0.4$. 
By contrast, the distribution of $\hti$ is approximately centered at $\tc$, coherent with \prop~\ref{prop:unification}\eqref{it:unification_tii}. 

\begin{figure}[!t]
\begin{center}
\includegraphics[width = .8\textwidth]{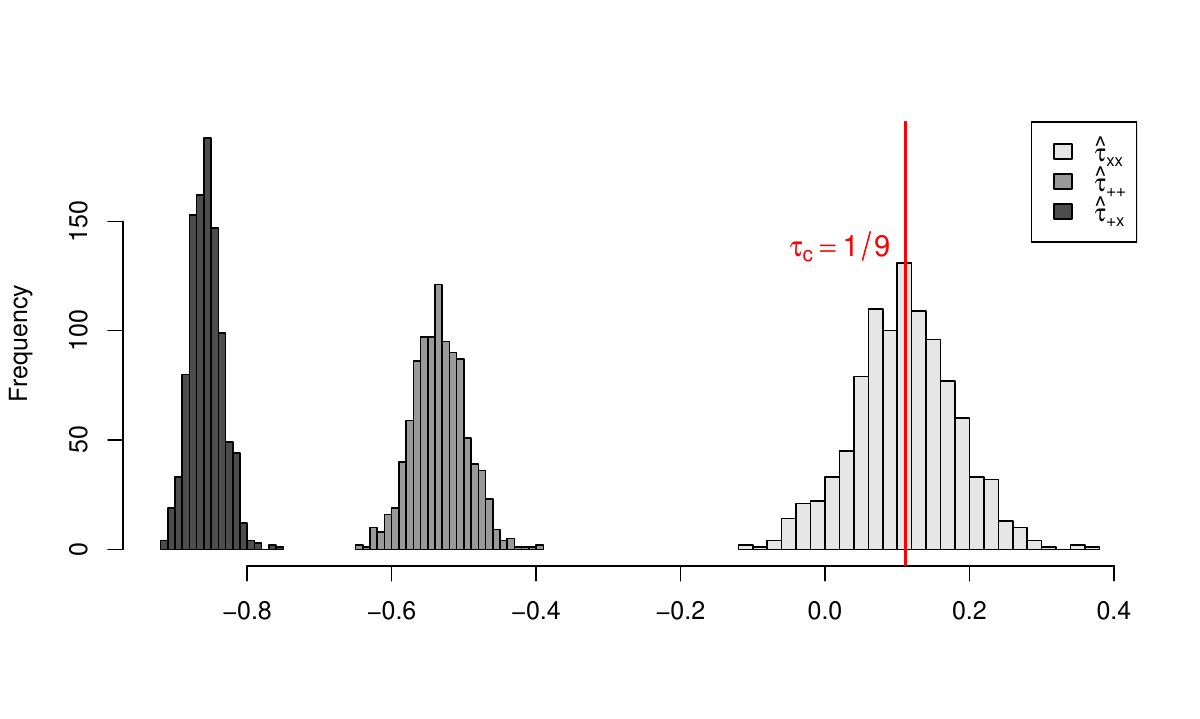}
\end{center}
\caption{\label{fig:LATE} Distributions of $\hti$, $\hta$, and $\htia$ over 1,000 replications.}
\end{figure}

\section{Extensions}\label{sec:ext}
\subsection{Extension to partially interacted 2{\normalsize SLS}}\label{sec:partial}
Let $X_i$ denote the covariate vector that ensures the conditional validity of the IV according to \assmiv. 
The interacted \tsls\ in Definition \ref{def:2sls_i} interacts $D_i$ with the full set of $X_i$ to accommodate \hte. 
When \hte\ is of interest or suspected only for a subvector $V_i$ of $X_i$, a natural modification is to interact $D_i$ only with $V_i$ in fitting the \tsls, formalized in Definition \ref{def:2sls_p} below. See \cite{resnjanskij2024can} for a recent application.

\begin{definition}[Partially interacted \tsls]\label{def:2sls_p}
Consider the \tsls\ regression of $Y_i$ on $(\di V_i, \xxi)$, instrumented by $(\zi V_i, \xxi)$ over $i = \ot{N}$: 
$$
\tslst(Y_i \sim \di V_i + \xxi \mid \zi V_i + X_i).
$$   
Let $\hbv$ denote the coefficient vector of $\di V_i$. 
\end{definition}

The partially interacted \tsls\ in Definition \ref{def:2sls_p} generalizes the interacted \tsls, and reduces to the interacted \tsls\ when $V_i = X_i$. 
It is the \tsls\ procedure corresponding to the working model
\begineq\label{eq:wm_partial}
Y_i= D_i V_i^\T \beta_{DV} + X_i^\T \beta_X + \eta_i,
\endeq
where $\E(\eta_i \mid Z_i, X_i) = 0$. 
Under this model, $V_i^\T \beta_{DV}$ represents the covariate-dependent treatment effect. 
We establish below the causal interpretation of $\hbv$ under the LATE framework when \eqref{eq:wm_partial} is possibly misspecified. 

\paragraph{Causal estimand.} Analogous to the definition of $\bc$ in \eqref{eq:bc_def}, let $\proj_{U_i = \cc}(\tau_i \mid V_i) = V_i^\T \bcv $ denote the linear projection of $\tau_i$ onto $V_i$ among compliers, with 
\begina
\bcv  = \textup{argmin}_{b \in \mathbb R^Q}\E\Big\{ (\tau_i - \vit b)^2 \mid \uc\Big\}. 
\enda
Parallel to the discussion after \eqref{eq:bc_def}, $\vit \bcv $ is also the linear projection of the conditional LATE $\tc( X_i)$ on $V_i$ among compliers with 
$
\bcv  = \textup{argmin}_{b \in \mathbb R^Q}\E[ \{\tc(X_i) - \vit b\}^2 \mid \uc]$.
Thus, $\vit \bcv $ is the best linear approximation to both $\tau_i$ and $\tc(X_i)$ based on $V_i$ among compliers. 
We define $\bcv $ as the causal estimand for quantifying \hte\  \wrt\ $V_i$ among compliers. 

\paragraph{Properties of $\hb'_\textup{2\tiny SLS}$ for estimating $\bcv$.} 
Let $\bv$ denote the probability limit of $\hbv$. 
Proposition~\ref{prop:v} below extends \thmbi\ by characterizing level independence and providing a sufficient condition for identifying $\bcv$.

\begin{proposition}\label{prop:v}
\agv. 
Assume Assumption~\ref{assm:iv}. Then 
\begine[(i)]
\item\label{it:prop_v_level indep} $\bi'$ is level-independent if and only if 
\begineq\label{eq:cond_v}
\E(\zvi\mid \xxi) = e(X_i) V_i \ \ \text{is linear in $X_i$.}
\endeq
\item\label{it:prop_v_bc} $\bi' = \bc'$ for any possible potential-outcome model if $e(X_i)V_i$ is linear in $V_i$.
\ende 
\end{proposition}

Proposition~\ref{prop:v} shows that partial interactions change both the target and the identifying restrictions. 
The target $\bc'$ is the projection of $\tau_i$ on $V_i$ among compliers, which need not equal the corresponding components of the full projection coefficient $\bc$. 
Condition~\eqref{eq:cond_v} can be weaker than Assumption~\ref{assm:ivps} because it requires linearity only of $\E(Z_iV_i\mid X_i)$. 
For identifying $\bc'$, it suffices that $e(X_i)V_i$ be linear in $V_i$, a stronger condition than \eqref{eq:cond_v}. Equivalently, $e(X_i)V_i = \proj(Z_iV_i\mid X_i) = \proj(Z_iV_i\mid V_i)$. 
Thus, if $W_i$ contributes to the linear projection of $Z_iV_i$ beyond $V_i$,
the coefficients on $D_iV_i$ from the specification omitting $D_iW_i$ need not
identify $\bc'$, analogous to omitted-variable bias.

As in \thmnece\eqref{it:thm_nece_sufficient}, two sufficient cases for $e(X_i)V_i$ to be linear in $V_i$, and hence for $\bi' = \bc'$, are:
 \begine[(a)]
 \item\label{it:sub_categorical_a}\textbf{Categorical covariates:} 
$V_i$ \cq\ and $\zi \indep \wi \mid \vi$. 
 \item\label{it:sub_random}\textbf{Random assignment of IV:} $Z_i \indep X_i$.
 \ende
Under~\eqref{it:sub_categorical_a}, the partially \its\ identifies $\bc'$ without interactions between $D_i$ and $\wi$.
However, \eqref{it:sub_categorical_a} does not imply \assmivps.
Although $e(X_i)V_i=e(V_i)V_i$ is linear in $V_i$ because $V_i$ is saturated, $\E(Z_iW_i \mid X_i) = e(V_i)W_i$ need not be linear in $X_i$.
Consequently, when case~\eqref{it:sub_categorical_a} holds but \assmivps\ fails, the partially \its\ identifies $\bc'$ for arbitrary potential-outcome models, whereas the fully \its\ is not guaranteed to identify $\bc$. 

In the random-assignment case \eqref{it:sub_random}, by contrast, \assmivps\ also holds, so the partially \its\ identifies $\bc'$ and the fully \its\ identifies $\bc$.

\subsection{Extension to weak monotonicity}\label{sec:wm_main}
We extend our results by relaxing \assmiv\eqref{it:mono} to \wmf\  \citep{kolesar2013estimation, sloczynski2022not, blandhol2026when}. 
\begin{assumption}[Weak monotonicity]\label{assm:wm}
The support of $\xxi$ partitions into $\mx_\cc$ and $\mx_\dd$ with
$\pr\{D_i(1) \ge D_i(0) \mid \xxi\} = 1$ for $X_i\in\mx_\cc$ and   $\pr\{D_i(1) \le D_i(0) \mid \xxi\} = 1$ for $X_i\in\mx_\dd$.
\end{assumption}

\assmivnomono\ imply that $\pcd > 0$. 
Following \cite{kolesar2013estimation} and \cite{sloczynski2022not}, define
\begineq\label{eq:tl}
 \tl = \E\left(\tau_i \mid \uicd \right), 
 \qquad
  \tl( X_i ) = \E\left(\tau_i \mid X_i,\, \uicd\right)
\endeq
as the LATE and conditional LATE in the absence of the strong monotonicity condition in \assmiv\eqref{it:mono}. 
The law of iterated expectations implies that $\tl =\E\left[\tl(X_i) \mid \uicd \right]$. 

Under \assmwm, when \(\pr(U_i=\dd)>0\), define
\begineqs
\tau_\dd(\xxi) = \E(\tau_i \mid \xxi, \ud) 
\endeqs
as the conditional average treatment effect among defiers, parallel to $\tc(X_i)$ and $\tl(X_i)$ in \eqref{eq:tcx} and \eqref{eq:tl}.  
Parallel to $\bc$, let 
\begineq\label{eq:bd_def}
\beginar{rcl}
\bd &=& \textup{argmin}_{b \in \mathbb R^K}\E\left\{  (\tau_i - \xit b )^2 \mid \ud\right\}\\
&=& \textup{argmin}_{b \in \mathbb R^K}\E\left[ \{\tau_\dd(X_i) - \xit b \}^2 \mid \ud\right]  
\endar
\endeq
denote the coefficient vector from the linear projection of $\tau_i$ onto $\xxi$ among defiers, quantifying the part of \hte\ among defiers that is linearly explained by $\xxi$. 
Let $\delta_i=D_i(1)-D_i(0)$, and define
\[
\bdiff
=
\{\mathbb E(\delta_iX_iX_i^\top)\}^{-1}
\mathbb E(\delta_iX_i\tau_i),
\]
assuming the inverse exists. 
By construction, $\bdiff$ is level-independent.
When both compliance types exist and the matrices $W_\cc = \E(X_i \xit \mid  \uc)$ and $W_\dd = \E(X_i \xit \mid  \ud)$ are invertible, this can be written as
\begineqs
\bdiff = \left\{\pi W_\cc - (1-\pi) W_\dd \right\}^{-1}\{\pi W_\cc  \cdot \bc - (1-\pi)W_\dd \cdot \bd\},
\endeqs
where  $\pi = \pr(\uc \mid \uicd)$.
This representation expresses $\bdiff$ as a matrix-weighted difference between $\bc$ and $\bd$. 
 
\prop~\ref{prop:bi_weak} below extends \thm~\ref{thm:bi}'s necessary-and-sufficient characterizations to \wmf. 

\begin{proposition}\label{prop:bi_weak}
Assume \assmivnomono\ and \assmwm. Then 
\begine[(i)]
\item\label{it:prop_bi_weak_nece} $\bi$ is level-independent if and only if \assmivps\  holds. 
\item\label{it:prop_bi_weak_suff} $\bi = \bdiff$ for all possible potential-outcome models if and only if \assmivps\ holds. 
\ende
\end{proposition}

Because strong monotonicity implies \wmf, 
the necessity in \prop~\ref{prop:bi_weak}\eqref{it:prop_bi_weak_nece} follows  from the necessity in \thmbi\eqref{it:thm_bi_nece}. 
Since \(\beta_{\mathrm{diff}}\) is level-independent, the necessity in \prop~\ref{prop:bi_weak}\eqref{it:prop_bi_weak_suff} follows from \prop~\ref{prop:bi_weak}\eqref{it:prop_bi_weak_nece}.
\prop~\ref{prop:bi_weak}\eqref{it:prop_bi_weak_suff} ensures that, under \assmivps, $\bi$ identifies $\bdiff$.
\thmbi\eqref{it:thm_bi_suff} is a special case when $\pi = 1$. 



\subsection{Connections to interacted {\normalsize OLS}}\label{sec:ext_ols_main}
Setting the IV equal to the treatment, $Z_i = D_i$, reduces the \its\ to interacted ordinary least squares (\ols). 
Therefore, our theory also applies to interacted \ols\ when treatment assignment is conditionally randomized given $\xxi$.
The discussion of interacted \ols\ in observational studies has largely focused on the identification of the average treatment effect or the average treatment effect on the treated \citep{imbens2009recent, kline2011oaxaca, chattopadhyay2023implied, chen2025potential}. 
Our results complement this literature by characterizing the identification of systematic treatment effect heterogeneity. 
We next use these interacted \ols\ results to construct \tsls\ estimators with $Z_i$ as an instrument for $D_i$.

\paragraph{Local average treatment effect (LATE).} 
Recall that $\txi  = (\xio , \ldots, X_{i,K-1})^\T$ denotes the vector of nonconstant covariates. 
Let \(\pi_\cc=\pr(U_i=\cc)\).
By \prop~\ref{prop:unification}\eqref{it:unification_tii} or \cite{chen2025potential}, 
under \assms~\ref{assm:iv}--\ref{assm:ivps}, 
the coefficients on $Z_i$ from the two regressions
\begineqs
\begin{aligned}
&\lmt(D_i\sim 1+Z_i+\txi+
Z_i\{\txi-\E(\txi)\}),\\
&\lmt(Y_i\sim 1+Z_i+\txi+
Z_i\{\txi-\E(\txi)\})
\end{aligned}
\endeqs
identify, respectively,
\begineqs
\E\{D_i(1)-D_i(0)\}=\pc, \qquad 
\E\{Y_i(D_i(1))-Y_i(D_i(0))\}=\pc\tc.
\endeqs
Their ratio therefore identifies $\tc$.
The numerical equivalence between \tsls\ and indirect least squares implies that this ratio equals the coefficient on $\di$ from 
\begineq\label{eq:ils_0}
\tslst\left(\yi\sim 1 + \di + \txi + Z_i\{\txi - \E(\txi)\} \mid 1 + \zi + \txi + Z_i\{\txi - \E(\txi)\}\right). 
\endeq 
Unlike the \its, the \tsls\ in \eqref{eq:ils_0} includes $\zi\{\txi - \E(\txi)\}$ rather than $\di\{\txi - \E(\txi)\}$ in the second stage. 
It also centers $\txi$ at its marginal mean rather than at the complier means used in Definition~\ref{def:2sls_demeaned}, because each component regression targets an average effect of $Z_i$ over the whole population.

\paragraph{Local average treatment effect on the treated (LATT).} 
Similarly, let $\cxi = \txi - \E(\txi\mid \zi = 1)$ denote the covariates centered at their means among units with $\zi = 1$. 
\cite{kline2011oaxaca} showed that, in the least squares regression of $\yi$ on $(1, Z_i, \cxi^\T, \zi\cxi^\T)^\T$, the coefficient on $\zi$ identifies $\E\{\yi(\dio) - \yi(\diz) \mid \zi = 1\}$, the average effect of $\zi$ among units with $\zi = 1$, provided that the odds $e(\xxi) / \{1 - e(\xxi)\}$ are linear in $\xxi$. 

Applying the same result to $D_i$,
the coefficient on $\zi$ from the corresponding
interacted \ols\ regression of $\di$ on $(1, Z_i, \cxi^\T, \zi\cxi^\T)^\T$ identifies
\begineqs
\E\{D_i(1)-D_i(0)\mid\zi=1\}
=\pr(\uc\mid\zi=1).
\endeqs
The ratio of these two coefficients
therefore identifies
\[
\tau_{\cc\mid Z=1}
=\E(\tau_i\mid\uc,\zi=1).
\]
By the numerical equivalence between \tsls\
and indirect least squares, this ratio equals
the coefficient on $\di$ from
\begineq\label{eq:ils}
\tslst(
\yi\sim 1+\di+\cxi+\zi\cxi
\mid
1+\zi+\cxi+\zi\cxi).
\endeq
Because $D_i=Z_i$ among compliers,
$\tau_{\cc\mid Z=1}
=\E(\tau_i\mid\uc,\di=1)$,
commonly referred to as the local average
treatment effect on the treated (LATT).
Like \eqref{eq:ils_0}, the \tsls\
in \eqref{eq:ils} differs from the \its\
in that it includes $\zi\cxi$ rather than
$\di\cxi$ in the second stage.
 
\bibliographystyle{Chicago}
\bibliography{refs_iv-x}

\newpage
\setcounter{equation}{0}
\setcounter{section}{0}
\setcounter{figure}{0}
\setcounter{example}{0}
\setcounter{proposition}{0}
\setcounter{corollary}{0}
\setcounter{theorem}{0}
\setcounter{table}{0}
\setcounter{condition}{0}
\setcounter{definition}{0}
\setcounter{assumption}{0}
\setcounter{lemma}{0}
\setcounter{remark}{0}

\renewcommand {\thedefinition} {S\arabic{definition}}
\renewcommand {\theassumption} {S\arabic{assumption}}
\renewcommand {\theproposition} {S\arabic{proposition}}
\renewcommand {\theexample} {S\arabic{example}}
\renewcommand {\thefigure} {S\arabic{figure}}
\renewcommand {\thetable} {S\arabic{table}}
\renewcommand {\theequation} {S\arabic{equation}}
\renewcommand {\thelemma} {S\arabic{lemma}}
\renewcommand {\thesection} {S\arabic{section}}
\renewcommand {\thetheorem} {S\arabic{theorem}}
\renewcommand {\thecorollary} {S\arabic{corollary}}
\renewcommand {\thecondition} {S\arabic{condition}}
\renewcommand {\thepage} {S\arabic{page}}
\renewcommand {\theremark} {S\arabic{remark}}

\renewcommand{\theHequation}{supp.\arabic{equation}}

\setcounter{page}{1}

\spacingset{1.5}
\addtolength{\textwidth}{2in}%
\addtolength{\textheight}{.5in}%
\addtolength{\topmargin}{-.4in}%

\begin{center}
\bf \Large 
Supplementary Material
\end{center}

Section~\ref{sec:lemma} presents supporting lemmas.
Section~\ref{sec:proof} proves the results in the main paper.

\paragraph{Notation.} For random vectors $Y \in \mathbb R^p$ and $X \in \mathbb R^q$ with finite second moments and invertible $\E(X X^\T)$, write
\begineqs
\proj(Y\mid X) = \E(Y X^\T)\{\E(X X^\T)\}^{-1}X,\qquad 
\res(Y\mid X) = Y - \proj(Y\mid X)
\endeqs
for the linear projection of $Y$ onto $X$ and its residual, respectively. 

\agv.  
\assms~\ref{assm:rank}--\ref{assm:rank_sub} below make explicit the rank conditions for the interacted and partially \its\ specifications (Definitions~\ref{def:2sls_i} and \ref{def:2sls_p}), respectively. 
\begin{assumption}\label{assm:rank} 
The matrix 
$\E\left\{\beginp
\zxi\\
X_i
\endp \Big(D_iX_i^\T, X_i^\T\Big)\right\}$ is invertible. 
\end{assumption}

\begin{assumption}\label{assm:rank_sub} 
The matrix 
$\E\left\{\beginp
Z_iV_i\\
X_i
\endp \Big(D_iV_i^\T, X_i^\T\Big)\right\}
$ is invertible. 
\end{assumption}
Under \assm~\ref{assm:rank}, let
\begineqs
\hdxi = \proj( D_i X_i \mid \zxi, X_i),
\qquad
\tdxi=\res(\hdxi \mid \xxi)
\endeqs
denote the first-stage projection for the \its\ in Definition~\ref{def:2sls_i} and its residual after projecting onto $\xxi$. 
Under \assm~\ref{assm:rank_sub}, let
\begineqs
\hdvi = \proj( D_i V_i \mid \zvi, X_i), 
\qquad
\tdvi=\res(\hdvi \mid \xxi)
\endeqs
denote the analogous quantities for the partially \its\ in Definition~\ref{def:2sls_p}. 

Let $\delta_i = \dio - \diz$, which equals $0$ for always-takers and never-takers, $1$ for compliers, and $-1$ for defiers. 

\section{Lemmas}\label{sec:lemma}
\subsection{Projection and conditional expectation}
Lemmas~\ref{lem:proj}--\ref{lem:fwl} below summarize standard results. We omit the proofs.  
\begin{lemma}\label{lem:proj}
For random vectors $A$ and $B$, 
\begine[(i)]
\item\label{it:lem_proj_i} $\E(A \mid B) = \proj(A \mid B)$ if and only if $\E(A \mid B)$ is linear in $B$; 
\item\label{it:lem_proj_ii} if $A$ is scalar and $A \indep B$, then 
$ \E(AB\mid B)=\proj(AB\mid B) = \E(A)B$;
\item\label{it:lem_proj_iii} if $A$ is binary, then $\E(A \mid B)$ is constant implies that $A \indep B$. 
\ende
\end{lemma}

\begin{lemma}\label{lem:exactly m}
Let $A$ be an $m\times 1$ random vector with invertible $\E(AA^\T)$. Then 
\begine[(i)]
\item\label{it:exactly m sub} for any subvector of $A$, denoted by $B$,  $\E(BB^\T)$ is invertible. 
\item\label{it:exactly m i} 
$A$ takes at least $m$ distinct values.
\item\label{it:exactly m ii} If $A$ takes exactly $m$ distinct values, denoted by $a_1, \ldots, a_m \in \mathbb R^m$, then 
\item[-] the matrix $\Gamma = (a_1, \ldots, a_m) \in \mathbb R^{m\times m}$ is invertible, and $A = \Gamma(1(A = a_1),\ldots, 1(A = a_m))^\T$. 
\item[-] $AA^\T$ is componentwise linear in $A$. 
\item[-] for any random vector $C$, the conditional expectation $\E(C \mid A) $ is linear in $A$.  
\ende 
\end{lemma}

\begin{lemma}[Population Frisch--Waugh--Lovell (FWL)]\label{lem:fwl}
Consider a random variable $Y \in \mathbb R$, a $p_1 \times 1$ random vector $X_1$, and a $p_2 \times 1$ random vector $X_2$. 
Let $\tilde Y = \res(Y \mid X_2)$ and $\tilde X_1 = \res(X_1 \mid X_2)$ denote the residuals from the linear projections of $Y$ and $X_1$ on $X_2$, respectively.
Assume the relevant projection coefficients are unique.
The coefficient vector on $X_1$ in the linear projection of $Y$ onto $(X_1^\T, X_2^\T)^\T$ equals that on $\tilde X_1$ in the linear projection of either $Y$ or $\tilde Y$ onto $\tilde X_1$. 
\end{lemma}

\subsection{Lemmas for the interacted 2{\normalsize SLS}}
Lemma \ref{lem:nondegenerate_X} below gives the projection formulas and rank properties under Assumption \ref{assm:rank}. 

\begin{lemma}\label{lem:nondegenerate_X}
Let 
$C_1$ and $C_0$ denote the coefficient matrices from the linear projection of $\dxi$ onto $(\zxi,\xxi)$, and 
$C_2$ denote the coefficient matrix from the linear projection of $\zxi$ onto $\xxi$: 
\begineqs
\hdxi = \proj(D_iX_i\mid \zxi, X_i) = C_1 \zxi + C_0 X_i,
\qquad
\proj(\zxi \mid \xxi) = C_2 \xxi.
\endeqs 
Assume Assumption \ref{assm:rank}. Then: 
\begine[(i)]
\item\label{it:nondegenerate_X_X} $\E\left\{\beginp
\zxi\\
X_i
\endp \Big(\zxi^\T, X_i^\T\Big)\right\}$ and $\E(X_i \xt_i )$ are both invertible, with 
\begina
(C_1, C_0) &=& \E\Big\{ D_iX_i ( \zxi^\T, \xt_i ) 
 \Big\} \left[\E\left\{\beginp
\zxi\\
X_i \endp ( \zxi^\T, \xt_i ) \right\} \right]^{-1}, \\
 C_2 &=& \E(\zxi\xt_i ) \big\{\E(X_i \xt_i ) \big\}^{-1}.
 \enda 
\item\label{it:nondegenerate_X_K}$X_i $ takes at least $K$ distinct values. 
\item\label{it:nondegenerate_X_C1} $C_1$, $C_2$, and $I_K - C_2$ are all invertible. 
\ende 
\end{lemma}

\begin{proof}[\bf Proof of Lemma \ref{lem:nondegenerate_X}] 
Omit the subscript $i$ and define
\begineqs
H = \beginp ZX\\X \endp, \quad 
G = \beginp DX \\ X\endp, \quad 
\Gamma_1 = \E(G H^\T), \quad 
\Gamma_2
= 
\E(HH^\T).
\endeqs
Under Assumption~\ref{assm:rank}, $\Gamma_1$ is invertible.

\paragraph{\underline{Proof of Lemma \ref{lem:nondegenerate_X}\eqref{it:nondegenerate_X_X}}.}
Note that $\Gamma_2$ is positive semi-definite. 
If $\Gamma_2$ were degenerate, there would exist a nonzero vector $a$ such that $\E\{(a^\T H)^2\} = a^\T \Gamma_2 a = 0$. 
Then $H^\T a = 0$ almost surely, which in turn implies $$\Gamma_1 a = \E\{ 
G(H^\T a)\}  = 0.$$ 
This contradicts the invertibility of $\Gamma_1$. 
Thus $\Gamma_2$ is positive definite, and so is its principal submatrix $\exxt$. The coefficient formulas follow from linear projection.

\paragraph{\underline{Proof of Lemma \ref{lem:nondegenerate_X}\eqref{it:nondegenerate_X_K}}.}
The result follows from invertibility of $\exxt$ and  Lemma~\ref{lem:exactly m}\eqref{it:exactly m i}.

\paragraph{\underline{Proof of Lemma \ref{lem:nondegenerate_X}\eqref{it:nondegenerate_X_C1}}.}
Projecting $G$ onto $H$ gives $\proj(G \mid H) = \beginp C_1 & C_0\\ 0 & I_K\endp H$, with
\begineq\label{eq:proj_mat}
 \beginp C_1 & C_0\\ 0 & I_K\endp = \E(G H^\T) \{\E(HH^\T)\}^{-1} = \Gamma_1 \Gamma_2^{-1}.
\endeq
Since \(\Gamma_1\) and \(\Gamma_2\) are invertible, \eqref{eq:proj_mat} implies 
$
\det(C_1) = \det\beginp C_1 & C_0\\ 0 & I_K\endp = \det(\Gamma_1 \Gamma_2^{-1}) \neq 0
$, so $C_1$ is invertible. 

In addition, note that 
\begina
I_K - C_2 
&=&  \exx  \cdot \exxinv  - \E (Z X \xt ) \cdot \exxinv \\
& =& \E\left\{(1-Z ) X \xt \right\} \cdot \exxinv .
\enda
To show that $C_2$ and $I_K - C_2$
are invertible, it suffices to show that $\E(Z X \xt )$ and $\E\{(1-Z )X \xt \} $ are invertible. This is ensured by 
\begina
0 < \det(\Gamma_2) &=& \det\left\{\beginp \E(Z X \xt ) & \E(Z X \xt )\\ \E(Z X \xt ) & \E(X \xt )\endp\right\}\\
& =& \det\left\{\beginp \E(Z X \xt ) & \E(Z X \xt )\\ 0 & \E\left\{(1-Z )X \xt \right\}\endp\right\}\\
&
=&\det\Big\{ \E(Z X \xt ) \Big\}\cdot\det \Big[ \E\left\{(1-Z )X \xt \right\} \Big]. 
\enda
\end{proof}

The values of $\hbi$ and $\bi$ from the \its\ and $\bc$ defined in \eqref{eq:bc_def} all depend on the choice of the covariate vector $X_i$. 
To emphasize the dependence on covariate coding, 
write $\hbi(X_i)$, $\bi(X_i)$, and $\bc(X_i)$ for the coefficients computed using covariate vectors $\{X_i\in \mr^K: i = \ot{N}\}$. 
Lemma~\ref{lem:2sls_invar} below describes how these coefficients transform under an invertible linear transformation of $X_i$.

\begin{lemma}\label{lem:2sls_invar} 
For $K\times 1$ covariate vectors $\{X_i\in \mr^K: i = \ot{N}\}$ and an invertible $K\times K$ matrix $\Gamma$, 
whenever the coefficients are well defined, we have 
\begina
\hbi(\Gamma X_i) = (\Gamma^\T) ^{-1} \hbi(X_i), \quad \bi(\Gamma X_i) = (\Gamma^\T) ^{-1}\bi(X_i), \quad \bc(\Gamma X_i) = (\Gamma^\T)^{-1}\bc(X_i).
\enda 
\end{lemma}

Under \assmiv, if \assmivps\ holds or the interacted working model in \eqref{eq:wm} is correctly specified, the coefficient on $\di$ from the \its\ with $X_i = (1, \tilde X_i^\T)^\T$ identifies the first element of $\bc(X_i)$.
Lemma~\ref{lem:demeaned} below shows that centering the nonconstant covariates at their complier means makes this element equal the LATE. 

\begin{lemma}\label{lem:demeaned}
Assume \assmiv. 
For any $X_i = (1, \xio , \ldots, X_{i,K-1})^\T$, let $\xic = (1, \xio - \mu_1, \ldots, X_{i,K-1} - \mu_{K-1})^\T$ denote the population analog of $\xidm$. Then the first element of $\bc(\xic)$ equals $\tc$. 
\end{lemma} 

\begin{proof}[\bf Proof of Lemma \ref{lem:demeaned}]
Since the nonconstant components of $\xic$ have zero mean among compliers, 
the first normal equation implied by \eqref{eq:bc_def} implies
\begina
0 
&=& \E\left\{\tau_i - \xic^\T\bc(\xic)\mid \uic\right\} \\
&=& \tc - \E(\xic^\T\mid \uic)\bc(\xic) = \tc - [\bc(\xic)]_1.  
\enda
\end{proof} 

\subsection{Lemmas for the partially interacted 2{\normalsize SLS}}
Lemma~\ref{lem:nondegenerate_V} below extends Lemma~\ref{lem:nondegenerate_X} to the partially \its. We omit its proof. 

\begin{lemma}\label{lem:nondegenerate_V}
Assume Assumption~\ref{assm:rank_sub}. Then
\begine[(i)]
\item\label{it:nondegenerate_V_Q} $
\E\left\{\beginp
Z_iV_i \\
X_i
\endp \Big(Z_iV_i ^\T, \xit\Big)\right\}$, $\E(\xxit)$, and $\E( \vvit )$ are all invertible, with 
\begineq\label{eq:c1}
\hdvi = \proj( D_iV_i \mid Z_iV_i , X_i) = C_1 Z_iV_i + C_0 X_i, \qquad 
\proj(Z_iV_i \mid V_i) = C_2 V_i 
\endeq 
with 
\begineqs
\beginar{rcl}
(C_1, C_0) &=& \E\Big\{ 
 D_iV_i ( Z_i \vit , \xit ) 
 \Big\} \left[ \E\left\{\beginp
Z_iV_i \\
X_i
\endp \Big(Z_i \vit , \xit  \Big)\right\}\right]^{-1},\\
 C_2 &=& \E(Z_i \vvit )\{\E( \vvit )\}^{-1}.
 \endar
\endeqs
\item\label{it:nondegenerate_V_K} $X_i$ takes at least $K$ distinct values. $V_i$ takes at least $Q$ distinct values.  
\item\label{it:nondegenerate_V_C1} $C_1$,  $C_2$, and $I_Q - C_2$ are all invertible.
\ende 
\end{lemma}

Recall that $\bi'$ denotes the probability limit of $\hbi'$ from the partially \its, and that $\hdvi = \proj(D_iV_i\mid Z_iV_i, X_i)$ and $\tdvi = \res(\hdvi \mid \xxi)$.  

\begin{lemma}\label{lem:bi}
Assume~\assm~\ref{assm:rank_sub}. Then 
\begineqs
\bi' = \left\{\E\left(\tdvi\cdot \tdvi^\T\right) \right\}^{-1}\E\left(\tdvi \cdot Y_i\right),
\endeqs  
where $\E ( \tdvi \cdot \tdvi^\T) = \E\{\tdvi \cdot (D_i V_i)^\T\} $. 
\end{lemma}

\begin{proof}[Proof of \lem~\ref{lem:bi}]
Standard theory implies that $\bi'$ equals the coefficient vector on $\hdvi$ in $\proj(Y_i \mid \hdvi, X_i)$. 
The expression of $\bi'$ follows from \lem~\ref{lem:fwl}~(population FWL). 

To establish the second identity, write
\beginy\label{eq:dvi_decomp}
D_i V_i- \tdvi &=& \left(D_i V_i -\hdvi\right) + \left(\hdvi - \tdvi\right)\nonumber\\
&=& \res(D_iV_i \mid \zvi, X_i) + \proj(\hdvi \mid X_i).
\endy
By definition, $\tdvi  = \hdvi - \proj(\hdvi\mid X_i)$ is a linear combination of $(\zvi, \xxi)$ and is orthogonal to $\xxi$.
Projection orthogonality gives
\begineq\label{eq:dvi_proj}
\E\left\{ \tdvi \cdot \res(D_iV_i \mid \zvi, X_i)^\T\right\} = 0, \qquad  
\E\left\{ \tdvi \cdot \proj(\hdvi \mid X_i)^\T\right\} = 0. 
\endeq
Multiplying \eqref{eq:dvi_decomp} by $\tdvi$ and taking expectation implies
$
\E\{ \tdvi \cdot (D_i V_i)^\T\} -\E( \tdvi \cdot \tdvi^\T)
=
\E[\tdvi \cdot \{(D_i V_i)^\T - \tdvi^\T\}]\oeq{\eqref{eq:dvi_decomp}+\eqref{eq:dvi_proj}}0$.
\end{proof}

\begin{lemma}\label{lem:bi_plim}
Assume \assms~\ref{assm:iv}\eqref{it:indep}--\eqref{it:relevance} and \ref{assm:rank_sub}. 
For any fixed $\beta \in \mbr^Q$, define 
\begineqs
\Delta_i = \yiz  +  \diz  (\tau_i -  \vit\beta), \qquad \epi = \dti  (\tau_i -  \vit\beta). 
\endeqs 
Recall that $e(X_i) = \E(Z_i \mid X_i)$. Define
\begineq\label{eq:b1_b2_star}
\beginar{rcl}
B_1(\beta) &=& \E\Big[  \big\{e(X_i) V_i - \proj(Z_iV_i\mid X_i) \big\} \cdot \big\{\E(\Delta_i \mid X_i) - \proj(\Delta_i \mid X_i)\big\} \Big],\\
B_2(\beta) &=& \E \Big[ e(X_i) \left\{ V_i - \proj(\zvi \mid X_i) \right\} \cdot \ep_i \mid \uicd \Big].
\endar
\endeq
Let 
\begineqs
B(\beta) = C_1\Big[ B_1(\beta) + B_2(\beta) \cdot \pcd  \Big] ,
\endeqs
where 
$C_1$ is the coefficient matrix on $\zvi$ in $\hdvi = \proj( D_iV_i \mid \zvi, X_i)$; cf.~\eqref{eq:c1}.
Recall that $\tdvi=\res(\hdvi \mid \xxi)$. 
Then 
\begineqs
\bi'= \beta + \{\E(\tdvi\cdot \tdvi^\T) \}^{-1}B(\beta). 
\endeqs
Further assume \assmwm, and define $c(x) = 1\{x\in\mathcal X_c\} - 1\{x\in\mathcal X_d\}$. 
Then 
\begineqs
B_2(\beta) 
= \E \Big[ e(X_i)  \left\{ V_i - \proj(\zvi \mid X_i) \right\}\cdot c(\xxi)\left\{\tlx - \vit\beta\right\} \mid \uicd \Big].
\endeqs
\end{lemma}
 
\begin{proof}[\bf Proof of \lem~\ref{lem:bi_plim}] 
Let $q_i = \proj(\zvi \mid \xxi)$. 
From \eqref{eq:c1} and the definition of $\tdvi$, 
\begineq\label{eq:s9_0}
\tdvi = C_1(\zvi - q_i). 
\endeq 
Since $D_i = \diz + \zi \dti$ and $Y_i = \yiz + \di \tau_i$, 
\begineq
Y_i - D_i\vit\beta 
=
\yiz 
+  \diz  (\tau_i -  \vit\beta)
+  \zi\dti  (\tau_i -  \vit\beta)
=
\Delta_i + Z_i \epi. \label{eq:Delta_ep_wm}
\endeq
Combining \eqref{eq:Delta_ep_wm} with \lem~\ref{lem:bi} gives
\beginy
\bi' -\beta 
&\oeqt{\lem~\ref{lem:bi}}&  
\left\{\E\left(\tdvi\cdot \tdvi^\T\right) \right\}^{-1}\E\left\{ \tdvi \cdot (Y_i - D_iV_i^\T \beta) \right\}\nnb\\
&\oeq{\eqref{eq:Delta_ep_wm}}& 
\left\{\E\left(\tdvi\cdot \tdvi^\T\right) \right\}^{-1}
 \Big\{\E(\tdvi \cdot \Delta_i) + \E(\tdvi \cdot \zi\epi)\Big\}. \label{eq:S9_c}
\endy

First, conditional exogeneity implies $Z_i \indep \Delta_i \mid \xxi$. Thus, by \eqref{eq:s9_0}, 
\begineq\label{eq:s9_1}
\E( \tdvi \cdot \Delta_i) 
\oeq{\eqref{eq:s9_0}}
C_1 \E\{(\zvi - q_i)\cdot \Delta_i\}
=
C_1\E\Big[\{ e(X_i) V_i - q_i\} \cdot \E(\Delta_i \mid X_i) \Big]
= C_1B_1(\beta),
\endeq
where the last equality follows from projection orthogonality:
\begineqs
\E\Big[\{e(X_i)V_i-q_i\} \cdot \proj(\Delta_i\mid X_i)\Big]
=\E\Big\{(Z_iV_i-q_i) \cdot \proj(\Delta_i\mid X_i)\Big\}=0.
\endeqs

Similarly, conditional exogeneity implies $Z_i \indep \epi \mid \xxi$. 
Together with $\epi=0$ for always-takers and never-takers, this implies
\begineq\label{eq:S9_b}
\beginar{rcl}
\E(\tdvi \cdot Z_i\epi)
&\oeq{\eqref{eq:s9_0}}& C_1\E\{Z_i(V_i-q_i)\epi\}\\
&=& C_1\E\{e(X_i)(V_i-q_i)\epi\}
= C_1B_2(\beta)\cdot\pcd.
\endar
\endeq
Combining \eqref{eq:S9_c}--\eqref{eq:S9_b} proves the first result.

\paragraph{\underline{Second expression for $B_2(\beta)$.}} 
Under Assumption 3, $\delta_i=c(X_i)$ for compliers
and defiers, so
\begineqs
\E(\epi\mid X_i,U_i\in\{c,d\})
=c(X_i)\{\tau_{\mathrm{late}}(X_i)-V_i^\T\beta\}.
\endeqs
The second expression for $B_2(\beta)$ follows
by iterated expectations.
\end{proof}

\begin{lemma}\label{lem:sub}
\vlem.
Then $\E(Z_iV_i\mid X_i)$ is linear in $X_i$ and $\E(Z_i\mid X_i) \cdot \{V_i - \proj(Z_iV_i\mid X_i)\}$ is linear in $V_i$ if either of the following conditions holds: 
\begine[(i)]
 \item\label{it:sub_1} $V_i$ \cq\ and $\E(Z_i \mid X_i) = \E(Z_i\mid V_i)$. 
 \item\label{it:sub_2} $Z_i \indep X_i$. 
 \ende
\end{lemma}

\begin{proof}[\bf Proof of Lemma \ref{lem:sub}] We verify below the sufficiency of the two conditions, respectively. 

\paragraph*{Condition~\eqref{it:sub_1}.} 
By Lemma~\ref{lem:exactly m}\eqref{it:exactly m ii},
$\E(Z_iV_i\mid V_i)$ is linear in $V_i$.
Under condition~\eqref{it:sub_1}, 
$\E(Z_iV_i \mid X_i) = \E(Z_iV_i\mid V_i)$ is also linear in $V_i$, and hence in $X_i$.
By Lemma~\ref{lem:proj}\eqref{it:lem_proj_i}, we have
\begineqs
\proj(Z_iV_i\mid X_i)
=\E(Z_iV_i\mid X_i) 
=\E(Z_iV_i\mid V_i)
=\proj(Z_iV_i\mid V_i).
\endeqs
Moreover, 
\begineq\label{eq:linear in V}
\E(Z_i\mid X_i) \cdot \{V_i - \proj(Z_iV_i\mid X_i)\}
=
\E(Z_i\mid V_i) \cdot \{V_i - \proj(Z_iV_i\mid V_i)\}.
\endeq
Because $V_i$ takes exactly $Q$ distinct values, the right-hand side of \eqref{eq:linear in V} is linear in $V_i$ by Lemma~\ref{lem:exactly m}\eqref{it:exactly m ii}.

\paragraph*{Condition~\eqref{it:sub_2}.}
Let $p = \E(\zi)$. 
By $Z_i \indep X_i$ and Lemma~\ref{lem:proj}\eqref{it:lem_proj_i}, we have 
\begineq\label{eq:cond2_ss}
\E(Z_iV_i\mid X_i) = \E(Z_i\mid X_i) V_i = pV_i, 
\quad 
\proj(Z_iV_i\mid X_i) \oeqt{Lem.~\ref{lem:proj}\eqref{it:lem_proj_i}} \E(Z_iV_i\mid X_i)  = pV_i,
\endeq 
so that
\begineqs
\E(Z_i\mid X_i)\cdot \{ V_i - \proj(Z_iV_i\mid X_i) \} \oeq{\eqref{eq:cond2_ss}} p(1-p)V_i. 
\endeqs
\end{proof}

Define
\begina
\bcd' &=& \textup{argmin}_{b \in \mathbb R^Q}\E\left[  (\tau_i - \vit b )^2 \mid \uicd\right]\nnb\\
&=& \textup{argmin}_{b \in \mathbb R^Q}\E\left[ \{\tl(X_i) - \vit b\}^2 \mid \uicd\right]
\enda
as the common linear-projection coefficient of $\tau_i$ and $\tl(\xxi)$ onto $V_i$ in the pooled population of compliers and defiers. 
Lemma~\ref{lem:assm_po_app} below shows that under \wmf, partially \its\ recovers $\bcd'$ when the working model \eqref{eq:wm_partial} is correctly specified.

\begin{lemma}\label{lem:assm_po_app}
Assume \assms~\ref{assm:iv}\eqref{it:indep}--\eqref{it:relevance}, \ref{assm:wm}, and \ref{assm:rank_sub}.
If the working model in \eqref{eq:wm_partial} is correctly specified, then 
\begine[(i)]
\item\label{it:bi'_1} $\tl( X_i) = \vit\bdv$, with $\bcd' = \bdv$; 
\item\label{it:bi'_2} $\bi' = \bdv = \bcd'$.
\ende
\end{lemma} 

\begin{proof}[\bf Proof of \lem~\ref{lem:assm_po_app}]
When \eqref{eq:wm_partial} is correctly specified,
\begineqs
\E(Y_i\mid \zi, \xxi) = \E(D_i\mid \zi,\xxi) \vit \bdv + \xit\beta_X. 
\endeqs
Taking differences between $\zi = 1$ and $\zi = 0$ gives 
\begineqs
\E(Y_i \mid \zi=1, \xxi) - \E(Y_i \mid \zi = 0, \xxi) 
=
\Big\{\E(D_i \mid \zi=1, \xxi) - \E(D_i \mid \zi=0, \xxi)\Big\} \vit \bdv.
\endeqs
By the conditional Wald identity and \assmiv\eqref{it:relevance}, 
\begineq\label{eq:tcx_1}
\tlx = \dfrac{\E(Y_i \mid \zi = 1, \xxi) - \E(Y_i \mid \zi = 0, \xxi)}{\E(D_i \mid \zi = 1, \xxi) - \E(D_i \mid \zi = 0, \xxi)}
= \vit \bdv.
\endeq
The definition of $\bcd'$ then gives $\bcd' = \bdv$, proving \lem~\ref{lem:assm_po_app}\eqref{it:bi'_1}.

For \lem~\ref{lem:assm_po_app}\eqref{it:bi'_2}, recall that $\tdvi = \res(\hdvi\mid \xxi)$ is a function of $(\xxi,\zi)$ with
\begineq\label{eq:assm_po}
\E(\tdvi \cdot \xit) = 0, \qquad \E(\tdvi \cdot \eta_i) = \E\left\{\E(\tdvi \cdot \eta_i \mid \xxi,\zi)\right\} \oeq{\eqref{eq:wm_partial}} 0.
\endeq
This implies that 
\beginy
\tdvi\cdot Y_i 
&\oeq{\eqref{eq:wm_partial}}& \tdvi \cdot (\dvi)^\T \bdv + \tdvi\cdot (\xit\beta_X + \eta_i),\nnb\\
\E\left(\tdvi\cdot Y_i\right) &\oeq{\eqref{eq:assm_po}}&
\E\left\{\tdvi \cdot (\dvi)^\T \right\} \cdot\bdv
\oeqt{Lem.~\ref{lem:bi}}
\E\left(\tdvi\cdot \tdvi^\T \right) \cdot\bdv.
\label{eq:dxy_po}
\endy
Applying \lem~\ref{lem:bi} again yields $\bi' \oeq{\eqref{eq:dxy_po}+\textup{Lem.~\ref{lem:bi}}} \bdv$. 
\end{proof}

\subsection{A linear algebra lemma}\label{sec:lem_la}
\begin{lemma}\label{lem:la_x}
Let $X = (X_1, \ldots, X_m)^\T$ be an $m\times 1$ random vector. 
\begine[(i)]
\item\label{it:cor_i} If there exists an $m\times 1$ constant vector $c = (c_1, \ldots, c_m)^\T$ such that $ X X^\T c$ is linear in $(1, X)$, then $c^\T X$ takes at most $m+1$ distinct values.
\item\label{it:cor_ii}
If $ X X^\T c$ is linear in $(1, X)$ for every $m\times 1$ constant vector $c \in \mathbb R^m$, then $X$ takes at most $m+1$ distinct values. 
\ende
\end{lemma}

\begin{proof}[Proof of \lem~\ref{lem:la_x}]
Let $S=\operatorname{supp}(X)$. By continuity, any identity below
that holds almost surely also holds throughout $S$.

\medskip
\noindent
\paragraph{\lem~\ref{lem:la_x}\eqref{it:cor_i}.} 
By assumption, there exist $a\in\mathbb R^m$ and
$B\in\mathbb R^{m\times m}$ such that
\[
(c^\top X)X=a+BX.
\]
Define
\begineqs
M=\begin{pmatrix}
0 & c^\top\\
a & B
\end{pmatrix} \in \mathbb R^{(m+1)\times(m+1)}.
\endeqs
For every $x\in S$,
\[
M\begin{pmatrix}1\\x\end{pmatrix}
=(c^\top x)\begin{pmatrix}1\\x\end{pmatrix}.
\]
Since $(1,x^\top)^\top\ne0$, each $c^\top x$ is an eigenvalue of $M$. Thus $c^\top X$
takes at most $m+1$ distinct values.

\medskip
\noindent
\paragraph{\lem~\ref{lem:la_x}\eqref{it:cor_ii}.} 
Suppose, to the contrary, that $\operatorname{supp}(X)$
contains $m + 2$ distinct points
$x^{(1)},\ldots,x^{(m+2)}$.
For each $r\ne s$, the set of vectors $c$ satisfying
$c^\top x^{(r)}=c^\top x^{(s)}$ is a hyperplane through the origin. 
Since finitely many proper hyperplanes cannot cover $\mathbb R^m$, we can choose
$c\in\mathbb R^m$ outside their union.
Then $c^\top x^{(1)},\ldots,c^\top x^{(m+2)}$
are pairwise distinct, contradicting~\eqref{it:cor_i}.
\end{proof}
 
\section{Proofs of the results in the main paper}\label{sec:proof}
The \its\ in Definition~\ref{def:2sls_i} is a special case of the partially \its\ in Definition~\ref{def:2sls_p}, 
while \smf\ implies \wmf. 
We therefore first prove the results for the partially \its\ and \wmf\ in Section~\ref{sec:ext},
then establish the results for \secs~\ref{sec:hte}--\ref{sec:late} using these results where applicable. 

\subsection{Proof of \propv\ in Section~\ref{sec:partial}}
\begin{proof}[\bf Proof of \propv\eqref{it:prop_v_level indep}, sufficiency]
Suppose condition~\eqref{eq:cond_v} holds. 
Then \(\proj(Z_iV_i\mid X_i)=e(X_i)V_i\), so \(B_1(0)=0\) by \eqref{eq:b1_b2_star}. 
Under \assmiv, there are no defiers, so
\[
B_2(0)
=
\E\Big[
e(X_i)
\left\{
V_i-\proj(Z_iV_i\mid X_i)
\right\}
\tau_i
\mid U_i=\cc
\Big].
\]
By Lemma~\ref{lem:bi_plim}, $\beta'_{\tsls}$ depends on $\{Y_i(1),Y_i(0)\}$ only through $\tau_i$ and is therefore level-independent.
\end{proof}

\begin{proof}[\bf Proof of \propv\eqref{it:prop_v_level indep}, necessity.]
Let $q_i=\proj(Z_iV_i\mid X_i)$ and $h_i=\E(Z_iV_i\mid X_i)-q_i$.
By Lemma~\ref{lem:nondegenerate_V} and the definition of
$\widetilde{DV}_i$,
\begineqs
\widetilde{DV}_i = C_1(Z_iV_i-q_i),
\qquad
\E(\widetilde{DV}_i\mid X_i)=C_1h_i,
\endeqs
where $C_1$ is invertible.

Suppose $\beta'_{\mathrm{2sls}}$ is level-independent.
Shifting both potential outcomes by any square-integrable
function $f(X_i)$ leaves individual treatment effects and
$\widetilde{DV}_i$ unchanged. Thus, Lemma~\ref{lem:bi} implies
\begineqs
0=\E\{\widetilde{DV}_i \cdot f(X_i)\}
=C_1\E\{h_if(X_i)\}.
\endeqs
Since $C_1$ is invertible, choosing $f(X_i)$ to be each
coordinate of $h_i$ yields $\E(h_{ij}^2)=0$ for
$j=1,\ldots,Q$. Hence $h_i=0$ almost surely,
proving \eqref{eq:cond_v}.
\end{proof}

\begin{proof}[\bf Proof of \propv\eqref{it:prop_v_bc}]
Note that $\bc'$ is level-independent. 
\propv\eqref{it:prop_v_level indep} implies that \eqref{eq:cond_v} is necessary for $\bi' = \bc'$ across arbitrary potential-outcome models. 

For sufficiency, suppose $\E(Z_iV_i\mid X_i) = e(X_i) V_i$ is linear in $V_i$, denoted by $\E(Z_iV_i\mid X_i) = A V_i$.
Then $\E(Z_iV_i\mid V_i) = \E\{\E(Z_iV_i\mid X_i, V_i) \mid V_i\} = AV_i$ so that, by \lem~\ref{lem:proj} and $\proj(Z_iV_i\mid V_i) = C_2V_i$ from \eqref{eq:c1}, 
\begineq\label{eq:prop2_1}
  \E(Z_iV_i\mid X_i) = \E(Z_iV_i\mid V_i)
  \oeqt{\lem~\ref{lem:proj}} \proj(Z_iV_i\mid V_i) \oeq{\eqref{eq:c1}}  C_2V_i. 
\endeq 
Applying \lem~\ref{lem:proj} again gives 
\begineqs
q_i=\proj(Z_iV_i\mid X_i) \oeqt{\lem~\ref{lem:proj}}  \E(Z_iV_i\mid X_i) \oeq{\eqref{eq:prop2_1}} C_2V_i,
\endeqs
so that 
\begineq\label{eq:prop2_2}
e(X_i)(V_i-q_i)=(I_Q-C_2)C_2V_i.
\endeq 
By the definition of $\beta'_c$,
\begineq\label{eq:prop2_3}
\E\left\{V_i(\tau_i-V_i^\T\beta'_c)\mid \uic \right\}=0.
\endeq 
Consequently,
\begineqs
\beginar{rcl}
B_2(\beta'_c)
&=& \E\left[e(X_i)(V_i-q_i)(\tau_i-V_i^\T\bc')
\mid \uic\right]\\
&\oeq{\eqref{eq:prop2_2}}& (I_Q-C_2)C_2
\E\left\{V_i(\tau_i-V_i^\T\bc')\mid \uic\right\}\\
&\oeq{\eqref{eq:prop2_3}}& 0.
\endar
\endeqs
Condition~\eqref{eq:cond_v} also implies
$B_1(\beta'_c)=0$. \lem~\ref{lem:bi_plim}
therefore yields $\beta'_{\mathrm{2sls}}=\beta'_c$.

\end{proof}

\subsection{Results under \wmf\ in Section~\ref{sec:wm_main}}
 
\begin{proof}[\bf Proof of \prop~\ref{prop:bi_weak}]
The necessity of \assmivps\ for level independence follows from the outcome-shift argument in the proof of \propv\eqref{it:prop_v_level indep}, which does not require \smf. 
Since $\beta_{\mathrm{diff}}$ is level-independent, 
\assmivps\ is also necessary for $\beta_{\mathrm{2sls}}=\beta_{\mathrm{diff}}$
under arbitrary potential-outcome models.

By the definition of $\bdiff$,  
\begineq\label{eq:weak_signed_normal_0}
\mathbb E\left\{\delta_iX_i
(\tau_i-X_i^\top\beta_{\mathrm{diff}})\right\}=0. 
\endeq
Under \assmivnomono, $\pr(\uicd)> 0$. 
Since $\delta_i = 0$ for always-takers and never-takers, \eqref{eq:weak_signed_normal_0} implies
\begineq
\E\left[\delta_iX_i(\tau_i-X_i^\T\beta_{\mathrm{diff}})
\mid \uicd \right] = 0. \label{eq:weak_signed_normal}
\endeq
Under Assumption~\ref{assm:ivps},
$
e(X_i)X_i =\proj(Z_iX_i\mid X_i)=C_2X_i,
$
where
$C_2=\E(Z_iX_iX_i^\T)\{\E(X_iX_i^\T)\}^{-1}$.
Thus, \lem~\ref{lem:bi_plim} with $V_i=X_i$ and
$\beta=\beta_{\mathrm{diff}}$ gives
$B_1(\beta_{\mathrm{diff}})=0$ and
\begina
B_2(\beta_{\mathrm{diff}})
&=&\E\left[
(X_i-C_2X_i)\cdot e(X_i) \cdot 
\delta_i(\tau_i-X_i^\T\beta_{\mathrm{diff}})
\mid \uicd
\right]\\
&=&(I_K-C_2)C_2 \cdot 
\E\left[\delta_iX_i(\tau_i-X_i^\T\beta_{\mathrm{diff}})
\mid \uicd\right]\\
&\oeq{\eqref{eq:weak_signed_normal}}&0. 
\enda
Therefore, \lem~\ref{lem:bi_plim} implies
$\beta_{\mathrm{2sls}}=\beta_{\mathrm{diff}}$,
proving~\eqref{it:prop_bi_weak_suff}. Since $\beta_{\mathrm{diff}}$ is
level-independent, sufficiency in~\eqref{it:prop_bi_weak_nece} follows.
\end{proof}

\subsection{Proofs of the results in Section~\ref{sec:hte}}
\begin{proof}[\bf Proof of \thmbi]
\thmbi\ follows from \propv\ with $V_i = \xxi$. 
\end{proof}

\begin{proof}[\bf Proof of \thm~\ref{thm:nece}]

\noindent\textit{\thm~\ref{thm:nece}\eqref{it:thm_nece_sufficient}.}
The sufficiency of saturated categorical covariates follows from Lemma~\ref{lem:exactly m}\eqref{it:exactly m ii}.
The sufficiency of randomly assigned IV follows from $\E(\zxi \mid X_i) = \E(Z_i \mid X_i) X_i = \E(Z_i) X_i$. 

\medskip

\noindent\textit{\thm~\ref{thm:nece}\eqref{it:assm_ezx_necessary}.} 
For \(K=1\), both claims are immediate. Suppose \(K\ge2\).
Let $\txi = \txif$, so that $X_i = (1, \txi^\T)^\T$. 
Under \assmivps, both $\E(Z_i \mid X_i)$ and $\E(Z_i \txi \mid X_i)$ are linear in $X_i$. 
Therefore, there exist constants $a_0\in \mr$ and $a_X\in\mr^{K-1}$ such that 
\beginy\label{eq:ezx}
e(X_i) = \E(Z_i \mid X_i) = a_0 + \txi^\T a_X,
\endy 
and therefore
\beginy\label{eq:ezx_1:K-1}
\E(Z_i \txi \mid X_i) =\txi \cdot \E(\zi \mid X_i)  \overset{\eqref{eq:ezx}}{=} 
 \txi a_0 + \txi \txi ^\T a_X.
\endy
Since $\E(Z_i \txi \mid X_i)$ is linear in $X_i$, \eqref{eq:ezx_1:K-1} implies $ \txi \txi ^\T a_X$ is linear in $ (1, \txi )$.
By \lem~\ref{lem:la_x}\eqref{it:cor_i}, $\txi ^\T a_X$ takes at most $K$ distinct values, and so does $e(X_i)$ by \eqref{eq:ezx}.

In addition, if \assmivps\ holds for every linear function $e_a(\cdot)=a^\T x$, then $ \txi \txit a_X$ is linear in $ (1, \txi ) $ for every $a_X \in \mbr^{K-1}$. 
By \lem~\ref{lem:la_x}\eqref{it:cor_ii}, $\txi$---and therefore $X_i$---takes at most $K$ distinct values. 
Lemma~\ref{lem:nondegenerate_X}\eqref{it:nondegenerate_X_K} guarantees at least $K$ distinct values, so $X_i$ takes exactly $K$ distinct values. 
\end{proof}

\begin{proof}[\bf Proof of the results on saturated categorical covariates]

Recall from \eqref{eq:bc_def} that $
\bc = \{\E(X_i\xit\mid \uc )\}^{-1}  \E( X_i  \tau_i \mid \cp )$. 
Suppose $X_i = \dmid$.
By the normal equations for the complier projection in \eqref{eq:bc_def}, for each $k = \ot{K}$, 
\begineqs
0
= \E\Big\{
1(X_i^*=k)
(\tau_i-X_i^\T\bc)
\mid \uic\big\}
= p_k (\tau_{[k]\cc} -[\bc]_k),
\endeqs
where $p_k=\pr(X_i^*=k\mid \uic)>0$.
Thus, $
\bc=(\tau_{[1]\cc}, \tau_{[2]\cc}, \ldots,\tau_{[K]\cc})^\T$ and $\tc(X_i)= \xit\bc$. 
By Theorems~\ref{thm:bi}--\ref{thm:nece}, \(\bi=\bc\). Lemma~\ref{lem:2sls_invar} ensures that this identification result and the fitted-value identity remain valid under any invertible linear transformation of \(X_i\).
\end{proof}

\subsection{Proofs of the results in Section~\ref{sec:late}}
\begin{proof}[\bf Proof of the connection with \cite{tan2006regression} and \cite{frolich2007nonparametric}]

Let $\htwk={\hat\Delta_{Y,k}}/{\hat\Delta_{D,k}}$
denote the subgroup Wald estimator for category $k$. Define the saturated dummy vector $
X_i'=(1(\xis = K), 1(\xis =1), \ldots,1(\xis  =K-1))^\top$.
The \its\ using $X_i'$ decomposes into
$K$ category-specific, just-identified IV regressions. Hence,
\begineq\label{eq:s1}
\hbi(X_i')
=
\left(
\hat\tau_{\mathrm{wald},[K]},
\hat\tau_{\mathrm{wald},[1]},
\ldots,
\hat\tau_{\mathrm{wald},[K-1]}
\right)^\top.
\endeq
Recall that $
X_{i,\mathrm{dm}}
= (
1, 1(\xis = 1)-\hat\mu_1, \ldots, 1(\xis = K-1)-\hat\mu_{K-1})^\top$. 
Then $X_i'=\Gamma X_{i,\mathrm{dm}}$, where
\begineqs
\Gamma
=
\begin{pmatrix}
\hat\mu_K & -1_{K-1}^\top \\
\hat{\mu}_{1:(K-1)} & I_{K-1}
\end{pmatrix},
\qquad
\hat{\mu}_{1:(K-1)}
=
\left(
\hat\mu_1,\ldots,\hat\mu_{K-1}
\right)^\top,
\qquad
\hat\mu_K=1-\sum_{k=1}^{K-1}\hat\mu_k.
\endeqs
\lem~\ref{lem:2sls_invar} implies
\begineq\label{eq:s2}
\hbi(X_{i,\mathrm{dm}})
=
\Gamma^\top
\hbi(X_i').
\endeq
Combining \eqref{eq:s1} and \eqref{eq:s2}, the first element of
$\hbi(X_{i,\mathrm{dm}})$ is
\begineq\label{eq:ss3}
\hat\tau_{\times\times}
=
\sum_{k=1}^K
\hat\mu_k\htwk
=
\sum_{k=1}^K
\hat\mu_k
\frac{\hat\Delta_{Y,k}}{\hat\Delta_{D,k}}.
\endeq
Now estimate the complier category shares by
\begineq
\hat\mu_k
=
\frac{
\hat{\pr}(X_i^*=k)\hat\Delta_{D,k}
}{
\sum_{j=1}^K
\hat{\pr}(X_i^*=j)\hat\Delta_{D,j}
},
\qquad
k=1,\ldots,K.
\label{eq:ss4}
\endeq
Substituting \eqref{eq:ss4} into \eqref{eq:ss3} gives
\begineqs
\hat\tau_{\times\times}
=
\sum_{k=1}^K
\frac{
\hat{\pr}(X_i^*=k)\hat\Delta_{D,k}
}{
\sum_{j=1}^K
\hat{\pr}(X_i^*=j)\hat\Delta_{D,j}
} \cdot 
\frac{\hat\Delta_{Y,k}}{\hat\Delta_{D,k}}
=
\frac{
\sum_{k=1}^K
\hat{\pr}(X_i^*=k)\hat\Delta_{Y,k}
}{
\sum_{j=1}^K
\hat{\pr}(X_i^*=j)\hat\Delta_{D,j}
}
=
\hat\tau_{\mathrm{nonpara}},
\endeqs
where $\hat\tau_{\mathrm{nonpara}}$ is the estimator common to \citet{tan2006regression} and \cite{frolich2007nonparametric}.  
\end{proof}

Let $\xic = (1, \xio  - \mu_1, \ldots, X_{i,K-1} - \mu_{K-1})^\T$ denote the population analog of $\xidm$.
Consistency of \(\hmuk\) and the rank condition imply that $\ti$ equals the first element of $\bi(\xic)$.  
\thm~\ref{thm:bi} and \lem~\ref{lem:2sls_invar} give $\bi(\xic) = \bc(\xic)$, where the first element of $\bc(\xic)$ equals $\tc$ by Lemma~\ref{lem:demeaned}.
This proves \prop~\ref{prop:unification}\eqref{it:unification_tii}. 
We prove below Proposition~\ref{prop:unification}\eqref{it:unification_taa}--\eqref{it:unification_tia}. 

\begin{proof}[\bf Proof of Proposition \ref{prop:unification}\eqref{it:unification_taa}]
The weighted-average representation follows from \citet[][Corollary 3.4]{sloczynski2022not}. When the first stage is correctly specified, write 
\begineqs
\E(D_i \mid \zi, \xxi) = a Z_i + b^\T X_i. 
\endeqs
This implies that $\pi(\xxi) 
= \E(D_i \mid \zi = 1, \xxi) -  \E(D_i \mid \zi = 0, \xxi) = a$ and 
\begineqs
\var\{\E(\di\mid \zi, \xxi)\mid \xxi \} = a^2 \var(Z_i \mid \xxi),
\endeqs
so that $w_+(X_i) =  w(X_i)$. 
\end{proof}

\begin{proof}[\bf Proof of Proposition \ref{prop:unification}\eqref{it:unification_tia}]
Let 
\begineq\label{eq:tia_tdi_def}
\cdi = \proj( D_i \mid \zxi, X_i),
\qquad 
\tdi = \res(\cdi \mid X_i). 
\endeq
By Lemma \ref{lem:fwl} (population FWL), 
\beginy\label{eq:fwl_add}
\tia = \dfrac{\E(Y_i \tdi )}{\E(\tdi^2 )}. 
\endy
Recall that $\delta_i = \dio - \diz$, and define
\begineqs
\Delta_i = Y_i(0) + \diz\left\{\tau_i - \tc( X_i)\right\}   ,
\qquad
\epi = \delta_i\{\tau_i - \tc(\xxi)\}.
\endeqs
It follows from $\di = \diz + \zi\delta_i$ that 
\begina
Y_i  
&=& Y_i(0) + \diz \left\{\tau_i - \tc( X_i)\right\}    + Z_i\delta_i \left\{\tau_i - \tc( X_i)\right\}  + D_i\tc( X_i)\\
&=& \Delta_i + Z_i\epi + D_i\tc( X_i). 
\enda
Thus, the numerator of \eqref{eq:fwl_add} decomposes as 
\beginy\label{eq:e_3}
\E(Y_i \tdi)= \E(\Delta_i  \tdi) + \E(Z_i \epi \tdi) +\E \left\{ D_i \tdi\cdot\tc( X_i)  \right\}.
\endy
We evaluate its three terms in turn.

\paragraph{\underline{Compute $\E(\Delta_i \tdi)$ in \eqref{eq:e_3}}.} 
Under \assmiv\eqref{it:indep}, $\Delta_i \indep \tdi \mid X_i$. 
Moreover, projection orthogonality implies
\begineq\label{eq:tsa_xi_x}
0 
= \E\left\{\proj(\Delta_i\mid X_i) \cdot \tdi\right\} 
= \E\left\{\proj(\Delta_i\mid X_i) \cdot \E(\tdi \mid X_i) \right\}.
\endeq 
Therefore, 
\beginy\label{eq:tsa_1}
 \E(\Delta_i \tdi)
&=&\E\left\{ \E(\Delta_i \mid X_i) \cdot \E( \tdi \mid X_i) \right\}\nnb\\
&\overset{\eqref{eq:tsa_xi_x}}{=}& \E\Big[\Big\{ \E(\Delta_i \mid X_i) - \proj(\Delta_i\mid X_i)\Big\}  \cdot \E(\tdi \mid X_i)\Big]\nnb\\
&\overset{\eqref{eq:tia_tdi_def}}{=}& \E\Big[\Big\{ \E(\Delta_i \mid X_i) - \proj(\Delta_i\mid X_i)\Big\}  \cdot \Big\{ \E(\hdi\mid X_i) - \proj(\hdi\mid X_i)\Big\}\Big] \nnb\\
&\equiv& B. 
\endy

\paragraph{\underline{Compute $\E(Z_i \epi\tdi)=0$ in \eqref{eq:e_3}}.}

Under \assmiv, conditional exogeneity gives $\epi \indep Z_i \tdi \mid X_i$, while the definition of $\tc(\xxi)$ implies 
\begineqs
\E(\epi\mid \xxi) = \E\big\{\tau_i - \tc( X_i) \mid \cp , X_i \big\}\cdot \pr(\cp \mid X_i)= 0. 
\endeqs
Therefore, 
\begineq\label{eq:tsa_2}
 \E(Z_i\epi\tdi) = \E\left\{\E(Z_i\epi \tdi \mid \xxi)\right\} 
 = \E\left\{\E(Z_i\tdi \mid \xxi) \cdot \E(\epi \mid \xxi)\right\} = 0.
\endeq

Plugging \eqref{eq:tsa_1} and \eqref{eq:tsa_2} into \eqref{eq:e_3} implies 
\begina\label{eq:tia_total}
\E(Y_i \tdi) = B +  \E\left\{ \E(D_i \tdi \mid \xxi)  \cdot \tc( X_i) \right\},
\enda
so that 
\begineq\label{eq:tia_simplified}
\tia \oeq{\eqref{eq:fwl_add}} \{\E(\tdi^2 )\}^{-1} B +   \E\left\{ \alpha(X_i)\cdot\tc( X_i)   \right\},
\endeq
where 
\begineqs
\alpha(X_i) = \dfrac{\E(D_i \tdi\mid \xxi)}{\E(\tdi^2 )}. 
\endeqs
First-stage projection orthogonality implies $\E(D_i\widetilde D_i)=\E(\widehat D_i\widetilde D_i)=\E(\widetilde D_i^2)$, so $\E\{\alpha(X_i)\}=1$.

\paragraph{\underline{Simplifications under \assmivps}.}
Write 
\begineq\label{eq:hdi}
\hdi = a^\T\zxi + b^\T\xxi.
\endeq
Under  \assmivps, 
%
\eqref{eq:hdi} implies that 
\begineq\label{eq:tia_tdi}
\tdi \overset{\eqref{eq:tia_tdi_def}+\eqref{eq:hdi}}{=}  a^\T \zxi - a^\T\E( \zxi\mid X_i) = (a^\T X_i)\left\{Z_i  - e(X_i)\right\}. 
\endeq
Thus, $\E(\tdi \mid \xxi) = 0$, which implies $B = 0$.  

To simplify $\alpha(X_i)$ from \eqref{eq:tia_simplified}, 
write $v(X_i) = \var(Z_i\mid \xxi)$. 
The denominator satisfies 
\beginy
 \E(\tdi^2 \mid X_i) &\overset{\eqref{eq:tia_tdi}}{=}& (a^\T\xxi)^2 \cdot \E\Big[\big\{ Z_i - \E(Z_i\mid X_i)\big\}^2\mid X_i \Big]= (a^\T\xxi)^2 \cdot v(X_i), \nonumber \\
\E(\tdi^2) &=& \E\left\{(a^\T\xxi)^2 \cdot v(X_i)\right\}. \label{eq:taa_tdi2}
\endy
In addition,  under conditional exogeneity and monotonicity, 
\beginy\label{eq:taa_edd01}
\E\Big[D_i\{Z_i-e(X_i)\}\mid X_i\Big]
&=& e(X_i) \cdot \E\Big[D_i\{Z_i-e(X_i)\}\mid Z_i = 1, X_i\Big]\nnb\\
&& + \{1-\exi\} \cdot \E\Big[D_i\{Z_i-e(X_i)\}\mid Z_i = 0, X_i\Big] \nnb\\
&=&  e(X_i)\{1-e(X_i)\}
\Big\{
\E(D_i\mid Z_i=1,X_i)
-\E(D_i\mid Z_i=0,X_i)
\Big\}\nnb\\
&=& v(X_i) \cdot \E\{\dio - \diz \mid X_i\}\nnb\\
&=& v(X_i) \cdot \pi(X_i). 
\endy
The numerator thus satisfies 
\beginy\label{eq:taa_dtd}
\E(D_i \tdi \mid X_i) 
\oeq{\eqref{eq:tia_tdi}} (a^\T\xxi) \cdot \E\Big[D_i\{Z_i-e(X_i)\}\mid X_i\Big]\oeq{\eqref{eq:taa_edd01}} (a^\T\xxi) \cdot \pi(X_i) \cdot v(X_i).
\endy
Equations~\eqref{eq:taa_tdi2} and \eqref{eq:taa_dtd} together ensure that
\beginy\label{eq:tia_ss_1}
\alpha(X_i) \oeq{\eqref{eq:tia_simplified}} \dfrac{\E(D_i \tdi \mid X_i)}{\E(\tdi^2)} \overset{\eqref{eq:taa_tdi2}+\eqref{eq:taa_dtd}}{=}
 \dfrac{ (a^\T X_i) \cdot \pi(X_i) \cdot v(X_i)}{ \E\left\{(a^\T X_i)^2 \cdot v(X_i)\right\}}.
\endy
We show below $a^\T X_i=\widetilde\pi(X_i)$,
the weighted linear projection defined before \prop~\ref{prop:unification}, to complete the proof. 

\paragraph{\underline{Compute $a^\T X_i$}.}
Under \assmivps, 
$e(X_i)X_i=\proj(Z_iX_i\mid X_i)$,
which is linear in $X_i$. Thus,
$\{Z_i-e(X_i)\}X_i$ is in the linear span
of $(Z_iX_i,X_i)$ and is therefore
orthogonal to $D_i-\hdi$.
Moreover, from \eqref{eq:hdi}, 
\begineqs
\E\Big[\hdi\{Z_i-e(X_i)\}\mid X_i\Big]
=(a^\T X_i) \cdot v(X_i).
\endeqs
Combining this with \eqref{eq:taa_edd01} gives 
\begina
0
=
\E\left[
\{Z_i-e(X_i)\}X_i(D_i-\hdi)
\right]
=
\E\Big[ v(X_i)X_i \left\{\pi(X_i)-a^\T X_i\right\}\Big].
\enda
Therefore, 
\begineqs
a =\left\{\E[v(X_i)X_iX_i^\T]\}^{-1}
\E\{v(X_i)X_i\pi(X_i)\right\},
\qquad
a^\T X_i=\widetilde\pi(X_i).
\endeqs
Consequently, 
\begineqs
\tau_{\times+} \oeq{\eqref{eq:tia_simplified}} \E\{\alpha(X_i) \cdot \tc(X_i)\} = \E\{w_\times(X_i) \cdot \tc(X_i)\},
\endeqs
where
\begineqs
w_\times(X_i)
=\frac{v(X_i)\cdot \widetilde\pi(X_i) \cdot \pi(X_i)}
{\E\{v(X_i)\cdot \widetilde\pi(X_i)^2\}}.
\endeqs
The projection normal equations also imply
$\E\{v(X_i) \cdot \widetilde\pi(X_i)\cdot \pi(X_i)\}
=\E\{v(X_i) \cdot \widetilde\pi(X_i)^2\}$, 
confirming that $\E\{w_\times(X_i)\}=1$.

Finally, if the interacted first stage is correctly
specified, then $\widetilde\pi(X_i)=\pi(X_i)$, and therefore
\begineqs
w_\times(X_i)
=\frac{v(X_i)\pi(X_i)^2}
{\E\{v(X_i)\pi(X_i)^2\}} = w(X_i).
\endeqs
\end{proof}

\end{document}